\documentclass[11pt]{article}
\usepackage{url}
\usepackage{color,fullpage,microtype}
\usepackage{graphicx,wrapfig,subfigure,titling}
\usepackage{amssymb,amsmath,amssymb,amsfonts}
\usepackage{amsthm}
\usepackage{textcomp}
\usepackage{multirow}
\usepackage{gensymb}
\usepackage{array}
\usepackage{hyperref}
\usepackage{times}
\usepackage[sort]{cite}
\usepackage[noend]{algorithmic}
\usepackage{algorithm}
\usepackage[T1]{fontenc}
\usepackage[utf8]{inputenc}
\usepackage{boiboites}

\newboxedtheorem[boxcolor=orange, background=blue!5, titlebackground=blue!20,
titleboxcolor = black]{theo}{}{anything}

\newtheorem{theorem}{Theorem}%[section]
\newtheorem{definition}[theorem]{Definition}

\newtheorem{corollary}[theorem]{Corollary}%[section]
\newtheorem{observation}[theorem]{Observation}
\newtheorem{lemma}[theorem]{Lemma}%[section]
\newcommand{\expec}[1]{\mathbb E\left [ #1 \right ]}
\newcommand{\var}[1]{\mathbb V\left [ #1 \right ]}

\newcommand{\prob}[1]{\mathbb P \left [ #1 \right ]}

\newcommand{\sou}{|U|}
\DeclareMathOperator{\polylog}{polylog}
\DeclareMathOperator{\poly}{poly}
\DeclareMathOperator{\matching}{match}
\DeclareMathOperator{\match}{match}

\DeclareMathOperator{\sample}{\mathsf{Sample}}
\DeclareMathOperator{\opt}{opt}
\newcommand{\nofrac}[2]{\frac{#1}{#2}}
\newcommand{\eat}[1]{}
\DeclareMathOperator{\vc}{vc}
\DeclareMathOperator{\hs}{hs}

\newcommand{\F}{{\cal F}}
\newcommand{\calP}{{\cal P}}
\newcommand{\calF}{{\cal F}}

\newcommand*\samethanks[1][\value{footnote}]{\footnotemark[#1]}
\renewcommand{\paragraph}[1]{\vspace{0.05in}\noindent{\bf \boldmath #1}}

\title{Kernelization via Sampling \\
with Applications to Dynamic Graph Streams}

\author{\normalsize
Rajesh Chitnis\thanks{
The Weizmann Institute of Science, Rehovot, Israel. Supported by a postdoctoral fellowship from I-CORE ALGO. Email: \texttt{rajesh.chitnis@weizmann.ac.il} }
\and \normalsize
Graham Cormode\thanks{Department of Computer Science, University of
  Warwick, UK. Supported in part by the Yahoo Faculty Research and Engagement Program
and a Royal Society Wolfson Research Merit Award. Email: \texttt{g.cormode@warwick.ac.uk.}
}
\and  \normalsize
Hossein Esfandiari\thanks{Department of Computer Science, University of Maryland. Supported in part by NSF CAREER award 1053605, NSF Grant CCF-1161626, ONR YIP award N000141110662,
DARPA/AFOSR grant FA9550-12-1-0423, and a Google Faculty Research award. Email: \texttt{\{hossein, hajiagha\}@cs.umd.edu}}
    \and \normalsize
 MohammadTaghi Hajiaghayi\samethanks[3]
    \and \normalsize
Andrew McGregor\thanks{University of Massachusetts Amherst. Supported by NSF CAREER Award CCF-0953754 and CCF-1320719 and a Google Faculty Research Award. Email: \texttt{\{mcgregor,svorotni\}@cs.umass.edu}}
    \and \normalsize
Morteza Monemizadeh\thanks{Computer Science Institute of Charles University,
Faculty of Mathematics and Physics, Prague, Czech Republic. Partially supported by the project 14-10003S of GA \v{C}R. Part of this work was done when the author was at Department of Computer Science, Goethe-Universit\"{a}t Frankfurt, Germany and supported in part by MO 2200/1-1. Email: \texttt{monemi@iuuk.mff.cuni.cz}}
    \and \normalsize
Sofya Vorotnikova\samethanks[4]
}

\date{}
\begin{document}
\maketitle
\thispagestyle{empty}

\begin{abstract}
In this paper we present a simple but powerful subgraph sampling primitive that is applicable in a variety of computational models including dynamic graph streams (where the input graph is defined by a sequence of edge/hyperedge insertions and deletions) and distributed systems such as MapReduce. In the case of dynamic graph streams, we use this primitive to prove the following results:
\begin{itemize}
\item {\em Matching:} First, there exists an $\tilde{O}(k^2)$ space algorithm that returns an exact maximum matching on the assumption the cardinality is at most $k$. The best previous algorithm used $\tilde{O}(kn)$ space where $n$ is the number of vertices in the graph
and we prove our result is optimal up to logarithmic factors. Our algorithm has $\tilde{O}(1)$ update time. 
Second, there exists an $\tilde{O}(n^2/\alpha^3)$ space algorithm that returns an $\alpha$-approximation for matchings of arbitrary size. (Assadi et al.~\cite{AssadiKLY15} showed that this was optimal and independently and concurrently established the same upper bound.) We  generalize both results for weighted matching. Third, there exists an $\tilde{O}(n^{4/5})$ space algorithm that returns a constant approximation in graphs with bounded arboricity. 
While there has been a substantial amount of work on approximate matching in insert-only graph streams, these are the first non-trivial results  in the dynamic setting. 

\item {\em Vertex Cover and Hitting Set:} There exists an $\tilde{O}(k^d)$ space algorithm that solves the minimum hitting set problem where $d$ is the cardinality of the input sets and $k$ is an upper bound on the size of the minimum hitting set. We prove this is  optimal up to logarithmic factors. Our algorithm has $\tilde{O}(1)$ update time. The case $d=2$ corresponds to minimum vertex cover.
\end{itemize}

Finally, we consider a larger family of parameterized problems (including $b$-matching, disjoint paths, vertex coloring among others) for which our subgraph sampling primitive yields fast, small-space dynamic graph stream algorithms. We then show lower bounds for natural problems outside this family. 
\end{abstract}

\clearpage
\setcounter{page}{1}

\section{Introduction}
Over the last decade, a growing body of work has considered solving graph problems in the data stream model. Most of the early work considered the insert-only variant of the model where the stream consists of edges being added to the graph and the goal is to compute properties of the graph using limited memory. Recently, however, there has been a significant amount of interest in being able to process dynamic graph streams where edges are both added and deleted from the graph\cite{AhnGM12a,AhnGM12b,AhnGM13,KapralovLMMS14,KapralovW14,GoelKP12,KutzkovP14a,BhattacharyaHNT15,AhnCGMW15,GuhaMT15}. These algorithms are all based on the surprising efficacy of using random linear projections, aka linear sketching, for solving combinatorial problems. Results include testing edge connectivity \cite{AhnGM12b} and vertex connectivity \cite{GuhaMT15}, constructing spectral sparsifiers \cite{KapralovLMMS14}, approximating the densest subgraph \cite{BhattacharyaHNT15}, correlation clustering \cite{AhnCGMW15}, and estimating the number of triangles \cite{KutzkovP14a}. For a recent survey of the area, see \cite{McGregor14}.

The concept of \emph{parameterized stream algorithms} was
explored by Chitnis et al.~\cite{ChitnisCHM15} and Fafianie and
Kratsch~\cite{kratsch}. Their work investigated a natural connection
between data streams and parameterized complexity. In parameterized
complexity, the time cost of a problem is analyzed in terms of not
only the input size but also other parameters of the input.
For example, while the classic vertex cover problem is NP complete, it can be solved via a simple branching algorithm
 in time $2^{k}\cdot \poly(n)$ where $k$ is the size of the optimal vertex cover.
An important concept in parameterized complexity is \emph{kernelization} in which the goal is to efficiently transform an instance of a problem into a smaller instance such that the smaller instance is a ``yes'' instance (e.g., has a solution of at least a certain size) iff the original instance was also a ``yes'' instance.
For more background on parameterized complexity and kernelization, see~\cite{DF99, FG06}.
Parameterizing the \emph{space} complexity of a problem in terms of the size of the output is a particularly appealing notion in the context of data stream computation. In particular, the space used by any algorithm that returns an actual solution (as opposed to an estimate of the size of the solution) is necessarily at least the size of the solution.

\paragraph{Our Results and Related Work.}
In this paper we present a simple but powerful subgraph sampling primitive that is applicable in a variety of computational models including dynamic graph streams (where the input graph is defined by a sequence of edge/hyperedge insertions and deletions) and distributed systems such as MapReduce. This primitive will be useful for both parameterized problems whose output has bounded size and for solving problems  where the optimal solution need not be bounded. In the case where the output has bounded size, our results can be  thought of as \emph{kernelization via sampling}, i.e., we sample a relatively small set of edges according to a simple (but not uniform) sampling procedure and can show that the resulting graph has a solution of size at most $k$ iff the original graph has an optimal solution of size at most $k$. We present the subgraph sampling primitive and implementation details in Section \ref{sec:basic:sample}.

\paragraph{Graph Matchings.}
Finding a large matching is the most well-studied graph problem in the data stream model
\cite{AhnG11BM,EpsteinLMS09,AG11,mcgregor2005b,zelke,CrouchS14,mcgregor2005,KapralovKS14,Kapralov13,GoelKK12,KonradMM12,KonradR13}. However, all of the existing single-pass stream algorithms are restricted to the insert-only case, i.e., edges may be inserted but will never be deleted. This restriction is significant: for example, the simple greedy algorithm using $\tilde{O}(n)$ space returns a $2$-approximation if there are no deletions. In contrast, prior to this paper no $o(n)$-approximation was known in the dynamic case when there are both insertions and deletions. Finding an algorithm for the dynamic case of this fundamental graph problem was posed as an open problem in the Bertinoro Data Streams Open Problem List \cite{sublinear_open_64}.

In Section \ref{sec:matching}, we prove the following results for computing matching in the dynamic model. Our first result is an $\tilde{O}(k^2)$ space algorithm that returns a maximum matching on the assumption that its cardinality is at most $k$.   Our algorithm has $\tilde{O}(1)$ update time. The best previous algorithm~\cite{ChitnisCHM15} was the folklore algorithm that collects $\max(\deg(u),2k)$ edges incident to each vertex $u$ and finds the optimal matching amongst these edges. This algorithm can be implemented in $\tilde{O}(kn)$ space where $n$ is the number of vertices in the graph. Indeed obtaining an algorithm with $f(k)$ space, for any function $f$, in the dynamic graph stream case remained as an important open problem~\cite{ChitnisCHM15}.
We can also extend our approach to maximum weighted matching.
Our second result is an  $\tilde{O}(n^2/\alpha^3)$ space algorithm that returns an $\alpha$-approximation for matchings of arbitrary size. For example, this implies an $n^{1/3}$ approximation using $\tilde{O}(n)$ space, commonly known as the \emph{semi-streaming} space restriction \cite{muthukrishnan2003,mcgregor2005}.
Our third result is an $\tilde{O}(n^{4/5})$ space algorithm that returns a constant approximation in graphs with bounded arboricity (such as planar graphs). This result builds upon an approach taken by Esfandiari et al.~\cite{EsfandiariHLMO15} 
for the problem on insert-only graph streams.

\paragraph{Vertex Cover and Hitting Set.}
We next consider the  problem of finding the \emph{minimum vertex cover} and its generalization, \emph{minimum hitting set}. The hitting set problem can be defined in terms of hypergraphs: given a set of hyperedges,  select the minimum set of vertices such that every hyperedge contains at least one of the selected vertices.
If all hyperedges have cardinality two, this is the vertex cover problem.

There is a growing body of work analyzing hypergraphs in the data stream model \cite{GuhaMT15,SahaG09,EmekR14,RadhakrishnanS11,Sun13,KoganK14}. For example, Emek and Ros{\'{e}}n\cite{EmekR14} studied the following set-cover  problem which is closely related to the hitting set problem: given a stream of hyperedges (without deletions), find the minimum subset of these hyperedges such that every vertex is included in at least one of the hyperedges. They present an $O(\sqrt{n})$ approximation streaming algorithm using $\tilde{O}(n)$ space along with results for covering all but a small fraction of the vertices. Another related problem is independent set since the minimum vertex cover is the complement of the maximum independent set. Halld{\'{o}}rsson et al.~\cite{HalldorssonHLS10} presented streaming algorithms for finding large independent sets but these do not imply a result for vertex cover in either the insert-only or dynamic setting.

In Section~\ref{sec:hittingset}, we present a $\tilde{O}(k^d)$ space
algorithm that finds the minimum hitting set  where $d$ is the 
cardinality of the input sets and $k$ is an upper bound on the
cardinality of the minimum hitting set. We prove the space use is
optimal and matches the space used by previous algorithms in the
insert-only model \cite{ChitnisCHM15,kratsch}. Our algorithms can be
implemented with $\tilde{O}(1)$ update time. The only previous results
in the dynamic model were by Chitnis et al.~\cite{ChitnisCHM15} and
included a $\tilde{O}(kn)$ space algorithm and a $\tilde{O}(k^2)$
space algorithm under a much stronger ``promise'' that the vertex
cover of the graph defined by any prefix of the stream may never
exceed $k$. Relaxing this promise remained as the main open problem of Chitnis et al.~\cite{ChitnisCHM15}.
In Section~\ref{sec:hittingset}, we also generalize our exact matching result to hypergraphs. In Section~\ref{sec:lowerbounds}, we show our result is also optimal.

\paragraph{General Family of Results.}
Finally, we consider a larger family of parameterized problems for which
our subgraph sampling primitive yields fast, small-space dynamic graph
stream algorithms. This result is presented in
Section~\ref{sec:general}, while lower bounds for various problems outside this family
are proved in Section~\ref{sec:lowerbounds}.

\subsubsection{Recent Work on Approximate Matching} Two other groups have independently and concurrently made progress on the problem of designing algorithms that \emph{approximate} the size of the maximum matching in the dynamic graph stream model \cite{Konrad15,AssadiKLY15}. These are just relevant to our second result on matching (Section~\ref{largematchings}). Specifically, Assadi et al.~\cite{AssadiKLY15} showed that it was possible to $\alpha$-approximate the maximum matching using $\tilde{O}(n^2/\alpha^3)$ space; this matches our result. Furthermore, they also showed that this was near-optimal. Konrad \cite{Konrad15} proved slightly weaker bounds.

\section{Basic Subgraph Sampling Technique}
\label{sec:basic:sample}

\paragraph{Basic Approach and Intuition.}
The inspiration for our subgraph sampling primitive is the following simple procedure for \emph{edge} sampling. Given a graph $G=(V,E)$ and probability $p\in [0,1]$, let $\mu_{G,p}$ be the distribution $E\cup \{\bot\}$ defined by the following process:
\begin{enumerate}
\item Sample each vertex independently with probability $p$ and let $V'$ denote the set of sampled vertices.
\item Return an edge chosen uniformly at random from the edges in the induced graph on $V'$. If no such edge exists, return $\bot$.
\end{enumerate}

The distribution $\mu_{G,p}$ has some surprisingly useful
properties. For example, suppose that the optimal matching in a graph
$G$ has size at most $k$. It is possible to show that this matching has the same size as the optimal matching in the graph formed by taking $O(k^2)$ independent samples from $\mu_{G,1/k}$. It is not hard to show that such a result would not hold if the edges were sampled uniformly at random.\footnote{To see this, consider a layered graph on vertices $L_1\cup L_2\cup L_3\cup L_4$ with edges forming a complete bipartite graph on $L_1\times L_2$, a complete bipartite matching on $L_2\times L_3$, and a perfect matching on $L_3\times L_4$. If $|L_1|=n\gg k$ and $|L_2|=|L_3|=|L_4|=k/2$ then the maximum matching has size $k$ and every matching includes all edges in the perfect matching on $L_3\times L_4$. Since there are $\Omega(nk)$ edges in this graph we would need $\Omega(nk)$ edges sampled uniformly before we find  the matching on $L_3\times L_4$.
} The intuition is that when we sample from $\mu_{G,p}$ we are less likely to sample an edge incident to a high degree vertex then if we sampled  uniformly at random from the edge set. For a large family of problems including matching, it will be advantageous to avoid bias towards edges whose endpoints have high degree.

Our subgraph sampling primitive essentially parallelizes the process of sampling from $\mu_{G,p}$. This will lead to more efficient algorithms in the dynamic graph stream model. The basic idea is rather than select a subset of vertices $V'$, we randomly partition $V$ into $V_1\cup V_2 \cup \ldots \cup V_{1/p}$. Selecting a random edge from the induced graph on any $V_i$ results in an edge distributed as in $\mu_{G,p}$. Sampling an edge on each $V_i$ results in $1/p$ samples from $\mu_{G,p}$ although note that the samples are no longer independent. This lack of independence will not be an issue and will sometimes be to our advantage. In many applications it will make sense to parallelize the sampling further and select a random edge between each pair, $V_i$ and $V_j$, of vertex subsets. For applications involving hypergraphs we select random edges between larger subsets of $\{V_1, V_2, \ldots, V_{1/p}\}$.

\paragraph{Sampling Data Structure.}
We now present the subgraph sampling primitive formally. Given an unweighted graph $G=(V,E)$, consider a ``coloring'' defined by a function $c:V\rightarrow [b]$. It will be convenient to introduce the notation:
\[V_c = \{v\in V: c(v)=c\}\]
and we will say that every vertex in $V_c$ has color $c$.
For $S\subset [b]$, we say an edge or hyper-edge $e$ of $G$
is \emph{$S$-colored} if $c(e)=S$ where $c(e)=\{c(u):u\in e\}$ is the set of colors used to color vertices in $e$. Given this coloring and a constant $d\geq 1$, let $G'=(V,E')$ be a random subgraph where
\[E'={\textstyle \bigcup_{S\subseteq [b] : |S|\leq d}} E_{S}\]
and $E_{S}$ contains a single edge chosen uniformly from the set of
$S$-colored edges (or $E_{S}=\emptyset$ if there are none). In the case of a weighted graph, for each distinct weight $w$, $E_S$ contains a single edge chosen uniformly from the set of $S$-colored edges with weight $w$. 

\begin{definition}
We define $\sample_{b,d,1}$ to be the distribution over subgraphs
generated as above
where $c$ is chosen uniformly at random from a family of pairwise
independent hash functions.
$\sample_{b,d,r}$ is the distribution over graphs formed by taking the union of $r$ independent graphs sampled from $\sample_{b,d,1}$.
\end{definition}

\paragraph{Motivating Application.}
As a first application to motivate the subgraph sampling primitive we again consider the problem of estimating matchings. We will use the following simple lemma that will also be useful in subsequent sections (the proof, along with other omitted proofs, can be found in Appendix \ref{sec:omitted}).
 
\begin{lemma}
\label{lem:markov:ineq}
Let $U\subseteq V$ be an arbitrary subset of $|U|=r$ vertices and let
$c:V\rightarrow [4r\epsilon^{-1}]$ be a pairwise independent hash function. Then with probability at least $3/4$, at least
$(1-\epsilon)r$ of the vertices in $U$ are hashed to distinct
values. Setting $\epsilon<1/r$ ensures all vertices are hashed to
distinct values with this probability.
\end{lemma}

Suppose $G$ is a graph with a matching $M=\{e_1, \ldots, e_k\}$ of size $k$. Let $G'\sim \sample_{b,2,1}$. By the above lemma, there exists $b=O(k^2)$, such that all the $2k$ endpoints of edges in $M$ are colored differently with constant probability. Suppose the endpoints of edge $e_i$ received the colors $a_{i}$ and $b_i$. Then $G'$ contains an edge in $E_{\{a_i,b_i\}}$ for each $i\in [k]$. Assuming all endpoints receive different colors, no edge in $E_{\{a_i,b_i\}}$ shares an endpoint with an edge in $E_{\{a_j,b_j\}}$ for $j\neq i$. Hence, we can conclude that $G'$ also has a matching of size $k$. In Section~\ref{sec:general}, we show that a similar approach can be generalized to a range of problems.
Using a similar argument there exists $b=O(k)$ such
that $G'$ contains a constant approximation to the optimum matching.
However, in Section~\ref{sec:matching}, we show that there exists
$b=O(k)$ such that with high probability graphs sampled from 
$\sample_{b,2,O(\log k)}$  preserve the size of the optimal matching exactly.

\subsection{Application to Dynamic Data Streams and MapReduce}

We now describe how the subgraph sampling primitive can be implemented in various computational models.

\paragraph{Dynamic Graph Streams.}
Let $S$ be a stream of insertions and deletions of edges of an underlying graph
$G(V,E)$. We assume that vertex set $V=\{1,2,\ldots, n\}$. We assume that the length of stream is polynomially related to $n$ and hence $O(\log|S|)=O(\log n)$.
We denote an undirected edge in $E$ with two endpoints $u,v\in V$ by
$uv$. For weighted graphs, we assume that the weight of an edge is specified when the edge is inserted and deleted and that the weight never changes.
The following theorem establishes that the sampling primitive can be efficiently implemented in dynamic graph streams.

\begin{theorem}
\label{thm:sample:distribution}
Suppose $G$ is a graph with $w_0$ distinct weights. It is possible to sample from $\sample_{b,d,r}$ with probability at least $1-\delta$ in the dynamic graph stream model using $\tilde{O}(b^d rw_0)$ space and $\tilde{O}(r)$ update time.
\end{theorem}

\paragraph{MapReduce and Distributed Models.}
The sampling distribution is naturally parallel, making it
straightforward to implement in a variety of popular models.
In MapReduce, the $r$ hash functions can be shared state
among all machines, allowing Map function to output each edge keyed by
its color under each hash function.
Then, these can be sampled from on the Reduce side to generate the
graph $G'$.
Optimizations can do some data reduction on the Map side, so that only
one edge per color class is emitted, reducing the communication cost.
A similar outline holds for other parallel graph models such as
Pregel.

\section{Matchings and Vertex Cover}
\label{sec:matching}
In this section, we present results on finding the maximum matching and minimum vertex cover of a graph $G$. 
We  use $\matching(G)$ to denote the size of the maximum (weighted or unweighted as appropriate) matching in $G$ and use $\vc(G)$ to denote the size of minimum vertex cover.

\subsection{Finding Small Matchings and Vertex Covers Exactly}
\label{sec:exact}
The main theorem we prove in this section is:

\begin{theorem}[Finding Exact Solutions]\label{thm:exact}
Suppose $\matching(G)\leq k$. 
Then, with probability $1-1/\poly(k)$, 
\[\matching(G')=\matching(G) ~\mbox{ and }
\vc(G')=\vc(G)\ , \]
where $G'=(V,E')\sim \sample_{100k,2,O(\log k)}$.
\end{theorem}

\paragraph{Intuition and Preliminaries.} 
To argue that $G'$ has a matching of the optimal size, it  suffices to
show that for every edge $uv\in G$ that is not in $G'$, there are a
large number of edges incident to one or both of $u$ and $v$ that is in $G'$. If this is the case, then it will still be possible to match at least one of these vertices in $G'$. 

To make this precise, let $U$ be the subset of  vertices with degree at least $10k$. Let $F$ be the set of edges in the induced subgraph on $V\setminus U$, i.e., the set of edges whose endpoints both have small degree.  
We will prove that with high probability, 
\begin{equation} \label{eq:magic}
\left (F\subseteq E'\right ) \quad \mbox{ and } \quad \left ( \forall u\in U~,~\deg_{G'}(u)\geq 5k \right )\ ,
\end{equation}
where $E'$ is the set of edges in $G'$.
Note that any sampled graph $G'$ that satisfies this equation has the property that for all edges  $uv\in G$ that are not in $G'$ we have  $\deg_{G'}(u)\geq 5k$ or $\deg_{G'}(v)\geq 5k$.

\paragraph{Analysis.}
The first  lemma establishes that it is sufficient to prove that  \eqref{eq:magic} holds with high probability.

\begin{lemma}\label{eq1suffices}
If $\matching(G)\leq k$ then \eqref{eq:magic} implies $\matching(G')=\matching(G)$ and $\vc(G')=\vc(G)$.
\end{lemma}

The next lemma establishes that \eqref{eq:magic} holds with the required probability.
\begin{lemma}
Eq.~\ref{eq:magic} holds with probability at least $1-1/\poly(k)$.
\end{lemma}
\begin{proof}
First note that $\matching(G)\leq k$ implies that there exists a vertex cover $W$ of size of most $2k$ because the endpoints of the edges in a maximum matching form a vertex cover.
Next consider $H\sim \sample_{100k,2,1}$. We will show that for any $e\in F$ and  $u\in U$, 
\[
\prob{e\in H}  >  1/2 ~~\mbox{ and }~~
\prob{\deg_H(u)\geq 5k}  \geq  1/2 \ .
\]
  It  follows that if $r=O(\log k)$ and $G' \sim \sample_{100k,2,r}$ then 
  \[  \prob{e\in G' \mbox{ and } \deg_{G'}(u) \geq 5k }\geq 1-1/\poly(k) \ .\] 
  We then take the union bound over the $O(k^2)$ edges in $F$ and the $O(k)$ vertices in $U$. The fact that $|F|= O(k^2)$ and $|U|=O(k)$ follows from the promises $\matching(G)\leq k$ and $\vc(G)\leq 2k$. In particular, the induced graph on $V\setminus U$ has a matching of size $\Omega(|F|/k)$ since the maximum degree is $O(k)$ and this is at most $k$. Since all vertices in $U$ must be in the minimum vertex cover, $|U|\leq 2k$.

\paragraph{To prove $\prob{e\in H}\geq 1/2$.}
Let the endpoints of $e$ be $x$ and $y$. Consider the pairwise hash function $c:[n]\rightarrow [b]$ that defined $H$ where $b=100k$. If $c(x)\neq c(y)$ and $c(x),c(y)\not \in A$ where
\[A=\{c(w):w\in (W\cup \Gamma(x)\cup \Gamma(y)), w\not \in \{x,y\}\} 
\quad \mbox{ where $\Gamma(\cdot)$ denotes the set of neighbors of a vertex}\ ,\] 
then $xy$ is the unique edge in $E_{\{c(x),c(y)\}}$ and is therefore in $H$. This follows because any edge in $E_{\{c(x),c(y)\}}$ must be incident to a vertex in $W$ since $W$ is a vertex cover. However, the only vertices in $V_{c(x)}$ or $V_{c(y)}$ that are in $W$ are one or both of $x$ and $y$ and aside from the edge $xy$ none of the incident edges on either $x$ or $y$ are in $E_{\{c(x),c(y)\}}$.
Since $b=100k$ and $
|A|\leq 2k+ 10k+10k=22k$,
\begin{eqnarray*}
\prob{xy \in H} 
 \geq  1-\prob{c(x)=c(y)}-
\prob{c(x)\in A}-\prob{c(y)\in A}\geq 1-1/b-2|A|/b> 1/2 \ . 
\end{eqnarray*}

\paragraph{To prove $\prob{\deg_{H}(u)\geq 5k}\geq 1/2$.}
Let $N_u$ be an arbitrary set of $10k$ neighbors of $u$. If
$c(u)\not \in \{c(w):w\in W\setminus \{u\} \}$ and there exist
different colors $c_1, \ldots ,c_{5k}$ such that each
$c_i \in \{c(v):v \in N_u\}\setminus \{c(w):w\in W\}$ then  these
color pairings are unique to their edges, and the algorithm returns at least $5k$ edges incident to $u$.
This follows since every edge has at least one vertex in $W$. 

First note that  $\prob{c(u) \in \{c(w):w\in W\setminus \{u\}\}}\leq 2k/b$.
By appealing to Lemma \ref{lem:markov:ineq}, with probability at least $3/4$, there are at least $6k$ colors used to color the vertices $N_u$. Of these colors, at least $5k$ are colored differently from vertices in $W$. Hence we find $5k$ edges incident to $u$ with probability at least $3/4-2k/b\geq 1/2$.
\end{proof}

\paragraph{Extension to Weighted Matching.}
We now extend the result of the previous section to the weighted case. 
The following lemma shows that it is possible to remove an edge $uv$ from a graph without changing the weight of the maximum weighted matching, 
if $u$ and $v$ satisfy certain properties.

\begin{lemma}\label{lem:removingwe}
Let $G=(V,E)$ be a weighted graph and let $G'=(V,E')$ be a subgraph with the property:
 \[\forall uv \in E\setminus E' ~,~ \deg_{G'}^{w(uv)}(u)\geq 5k ~\mbox{  or  }~\deg_{G'}^{w(uv)}(v)\geq 5k\ ,\] 
where $\deg_G^w(u)$ is the number of edges incident to $u$ in $G$ with weight $w$.
Then, $\matching(G)=\matching(G')$.
\end{lemma}

Consider a weighted graph $G$ and let $G'\sim \sample_{100k,2,O(\log k)}$. For each weight $w$, let $G_w$ and $G'_w$ denote the subgraphs consisting of edges with weight exactly $w$. By applying the analysis of the previous section to $G_w$ and $G'_w$ we may conclude that $G'$ satisfies the properties of the above lemma. Hence, $\match(G)=\match(G')$. To reduce the dependence on the number of distinct weights in Theorem \ref{thm:sample:distribution}, we may first round each weight to the nearest power of $(1+\epsilon)$ at the cost of incurring a $(1+\epsilon)$ factor error. If $W$ is the ratio of the max weight to min weight, there are  $O(\epsilon^{-1} \log W)$ distinct weights after the rounding.

\subsection{Finding Large Matchings Approximately}
\label{largematchings}
We next show our graph sampling primitive yields an approximation algorithm for estimating large matchings. 

\paragraph{Intuition and Preliminaries.} 
Given a hash function $c:V\rightarrow [b]$, we say an edge $uv$ is
colored $i$ if $c(u)=c(v)=i$. If the endpoints have different colors, we say the edge is \emph{uncolored}. The basic idea behind our algorithm is to repeatedly sample a set of colored edges with distinct colors. Note that a set of edges colored with different colors is a matching. We use the edges in this matching to augment the matching already constructed from previous rounds. In this section we require the hash functions to be $O(k)$-wise independent and, in the context of dynamic data streams, this will increase the update time by a $O(k)$ factor.

\begin{theorem}\label{thm:largematch}
Suppose $\matching(G)\geq k$. 
For any $1\leq \alpha \leq \sqrt{k}$ and $0<\epsilon\leq 1$, with probability $1-1/\poly(k)$, \[\matching(G')\geq \left (\frac{1-\epsilon}{2\alpha} \right )\cdot k\ , \]
where  $G'\sim \sample_{2k/\alpha,1,r}$ where $r=O(k\alpha^{-2}\epsilon^{-2} \log k)$.
\end{theorem}

\begin{proof}
Let $H_1, \ldots, H_r \sim \sample_{2k/\alpha,1,1}$ and let $G'$ be the union of these graphs.
Consider the greedy matching $M_r$ where $M_0=\emptyset$ and for $t\geq 1$, $M_t$ is the union of $M_{t-1}$ and  additional edges from $H_t$. We will show that if $M_{t-1}$ is small, then we can find many edges in $H_t$ that can be used to augment $M_{t-1}$.

Consider $H_t$ and suppose $|M_{t-1}|< \nofrac{1-\epsilon}{2\alpha} \cdot k$. Let $c:V\rightarrow [b]$ be the hash-function used to define $H_t$ where $b=\nofrac{2k}{\alpha}$. Let $U$ be the set of colors that are not used to color the endpoints of $M_{t-1}$, i.e., 
\[U=\{c\in [b]: \mbox{ there does not exist a matched vertex $u$ in $M_{t-1}$ with $c(u)=c$}\}\ .\]
and note that $|U| \geq b-2|M_{t-1}| \geq \nofrac{k}{\alpha}$.
For each $c\in U$, define the indicator variable $X_c$ where $X_c=1$ if there exists an edge $uv$ with $c(u)=c(v)=c$. 
We will find $X=\sum_{c\in U} X_c$ edges to add to the matching.

Since $\matching(G)\geq k$, there exists a set $k-2|M_{t-1}|>
k\epsilon$ vertex disjoint edges that can be added to $M_{t-1}$.
Let $p=\nofrac{\alpha}{2k}$ and observe that 
\[
\expec{X_c}\geq k\epsilon p^2 - {k\epsilon \choose 2} p^4 >
k\epsilon p^2/2 = \epsilon \cdot \nofrac{\alpha^2}{8k}
\]
Therefore, $\expec{X}\geq (\nofrac{k}{\alpha})\cdot  \epsilon \cdot \nofrac{\alpha^2}{8k} = \nofrac{\epsilon \alpha}{8}$.
Since $X_c$ and $X_{c'}$ are negative correlated, 
$\prob{X \geq E[X]/2}\geq 1-\exp\left (-\Omega \left
(\epsilon \alpha  \right ) \right ) \geq \Omega(\epsilon)$.
Hence,
with each repetition we may increase the size of the matching by at
least $\epsilon  \alpha/2$ with probability $\Omega(\epsilon)$.
After $O(k\alpha^{-2} \epsilon^{-2} \log k)$ repetitions the matching has size at least $\nofrac{1-\epsilon}{2\alpha} \cdot k$.
\end{proof}

By applying Theorem \ref{thm:largematch}  for all $k\in \{1,2,4,8,16,\ldots\}$ and appealing to Theorem \ref{thm:sample:distribution}, we establish:

\begin{corollary}\label{cor:semimatch}
There exists a $O(n\polylog n)$-space algorithm that returns an $O(n^{1/3})$-approximation to the size of the maximum matching in the dynamic graph stream model.
\end{corollary}

This result generalizes to the weighted case using the Crouch-Stubbs technique \cite{CrouchS14}. They showed that if we can find a $\beta$-approximation to the maximum \emph{cardinality} matching amongst all edges of weight greater than $(1+\epsilon)^i$ for each $i$, then we can find a $2(1+\epsilon)\beta$-approximation to the maximum weighted matching in the original graph.

\subsection{Matchings in Planar and Bounded-Arboricity Graphs }

We also provide an algorithm for estimating the size of the matching in a graph of bounded arboricity. Recall that a graph has \emph{arboricity} $\nu$ if its edges can be partitioned into at most $\nu$ forests.  Our result is as follows.

\begin{theorem}
\label{thm:esimate:planar:matching}
There exists a $\tilde{O}(\nu \epsilon^{-2}n^{4/5} \log \delta^{-1} )$-space dynamic graph stream algorithm that returns a $(5\nu+9)(1+\epsilon)^2$ approximation of $\match(G)$ with probability at least $1-\delta$ where $\nu$ is the arboricity of $G$. 
\end{theorem}

The basic idea is to generalize the approach taken by Esfandiari et al.~\cite{EsfandiariHLMO15} in the insert-only case. This can be achieved using sparse recovery sketches and our algorithm for small matchings.
See Appendix~\ref{sec:planar}.

\section{Hitting Set and Hypergraph Matching}
\label{sec:hittingset}

In this section we present exact results for hitting set and hypergraph matching. 
Throughout the section, let $G$ be a hypergraph where each edge has size exactly $d$ and $\hs(G)\leq k$. In the case where $d=2$, the problems under consideration are vertex cover and matching. Throughout this section we assume $d$ is a constant. 

\paragraph{Intuition and Preliminaries.} Given that the hitting set problem is a generalization of the vertex cover problem, it will be unsurprising that some of the ideas in this section build upon ideas from the previous section. However, the combinatorial structure we need to analyze for our sampling result goes beyond what is typically needed when extending vertex cover results to hitting set. We first need to review a basic definition and result about ``sunflower'' set systems  \cite{ErdosR60}.

\begin{lemma}[Sunflower Lemma \cite{ErdosR60}]
Let $\F $ be a collection of subsets of $[n]$. Then $A_1,\ldots, A_s\in \F $ is an \emph{$s$-sunflower} if $A_i\cap A_j=C$
for all $1\leq i< j\leq s$. We refer to $C$
as the \emph{core} of the sunflower and $A_i \setminus C$ as the \emph{petals}. If each set in $\F$ has size at most $d$ and $|\F |> d! k^d$, then $\F$ contains a $(k+1)$-sunflower. 
\end{lemma}

Let $s_G(C)$ denote the number of petals in a maximum sunflower in the graph $G$ with core $C$. We say a core is \emph{large} if $s_G(C)>ak$ for some large constant $a$ and $\emph{significant}$ if $s_G(C)> k$.
Define the  sets:

\begin{itemize}
\item  $U = \{C \subseteq V \mid s_G(C) > ak\}$ is the set of large cores.
\item  $F = \{D \in E \mid \forall C \in U, C \not \subseteq D\}$ is
the set of edges that do include a large core.
\item   $U' = \{C \in U \mid \forall C' \subset C, s_G(C') \leq k\}$ is the set of large cores that do not contain significant cores.
\end{itemize}

The sets $U$ and $F$ play a similar role to the sets of the same name in the previous section. For example, if $d=2$, then a large core corresponds to a high degree vertex. However, the set $U'$ had no corresponding notion when $d=2$ because a high degree vertex cannot contain another high degree vertex. 
The following bounds on $|F|$ and $|U'|$ are proved in the Appendix \ref{sec:omitted}.
  
\begin{lemma}\label{lemma:smallf}
$|F| = O(k^d)$
 and $|U'| = O(k^{d-1})$
\end{lemma}

The next  lemma shows that if a core $C$ is contained in a set $D$, then the set of other edges $D'$ that intersect $D$ at $C$ has a hitting set that a) does not include vertices in $C$ and b) has small size if $s_G(C)$ is small.

\begin{lemma} \label{small_hs}
For any two sets of vertices $C \subseteq D$, define 
\[M_{C,D} = \{D' \setminus C \mid D' \in E, D \cap D' = C\} \ . \]
Then $\hs(M_{C,D})\leq s_G(C) d$.
\end{lemma}

\paragraph{Hitting Set.}
For the rest of this section we let
$G' = (V,E') \sim \sample_{b,d,r}(G)$
where $b=O(k)$, $d$ is the cardinality of the largest hyperedge, and $r=O(\log k)$. Let $W$ be a minimum hitting set of $G$.

\begin{theorem}
Suppose $\hs(G)\leq k$. With probability $1-1/\poly(k)$, 
$\hs(G')=\hs(G)$.
\end{theorem}

\begin{proof}
For each significant core $C$ there has to be at least one vertex from the hitting set in $C$. Since all large cores are significant, $\hs(G) = \hs(U \cup F)$. If $C \in U$ has a subset $C' \subset C$ such that $s_G(C') > k$, then there is at least one vertex from the hitting set in $C'$ and it also hits $C$. Thus, we only need to find significant cores that  do not contain other significant cores. Such sunflowers with more than $ak$ petals will be found according to Lemma \ref{lem_u}. Sunflowers with at most $ak$ petals will be found as a part of set $F$ according to Lemma \ref{lem_f}. 
\end{proof}

\begin{lemma}\label{lem_u}
$\prob{s_{G'}(C) > k \textrm{ for all } C \in U'} \geq 1-1/\poly(k)$.
\end{lemma}

\begin{proof}
Fix an arbitrary core $C\in U'$. Consider $H\sim \sample_{b,d,1}$ and let $c:[n]\rightarrow [b]$ be the coloring that defined $H$. We need to identify sets $S_1, S_2, \ldots S_{k+1}\subset [b]$ each of size $d$ with the following three properties:
\begin{enumerate}
\item All edges that are $S_i$-colored contain $C$
\item There is at least one $S_i$-colored edge.
\item If $D$ is $S_i$-colored and $D'$ is $S_j$-colored then $(D\setminus C)\cap (D'\setminus C)=\emptyset$.
\end{enumerate}
Let $D_1, D_2, \ldots, D_{k+1}$ be any set of edges where $D_i$ is $S_i$-colored. Then these sets form a sunflower of size $k+1$ on core $C$. 
 It will suffice to show that there exists such a family $S_1, S_2, \ldots S_{k+1}\subset [b]$ with probability at least $1/2$ because repeating the process $O(\log k)$ times will ensure that such a family exists with high probability. The result then follows by taking the union bound over all $C\in U'$ since $|U'| = O(k^{d-1})$.
 
\paragraph{Property 1.} We say $S\in [b]$ is \emph{good} if all $S$-colored edges contain $C$. 
We first define a set of vertices $A$ such that all edges disjoint from $A$ include $C$. Then any $S$ such that $S\cap c(A)=\emptyset$ will be good since if $c(D)=S$ for some edge $D$ then
$S\cap c(A)=\emptyset \Rightarrow 
c(D)\cap c(A)=\emptyset \Rightarrow 
D\cap A =\emptyset$,
and so $C\subseteq D$.
Let
\[
A=( W\setminus C) \cup \bigg({\textstyle \bigcup_{C' \subset C}} \hs(M_{C',C}) \bigg) \ .
\]
where $W$ is a minimum hitting set and, by a slight abuse of notation, we use $\hs(M_{C',C})$ to denote a minimum hitting set of $M_{C',C}$. Note that $\hs(M_{C',C})$ does not include any vertices in $C$. Since $W$ is a hitting set, all edges that do not intersect $W\setminus C$ must intersect with $C$. But all edges that intersect with only a subset of $C$, say $C'$, must intersect with $\hs(M_{C',C})$. Hence $A$ has the claimed property.

\paragraph{Properties 2 and 3.}
Next, let $\calP$ be a set of petals in a sunflower with core $C$ that do not intersect with $A$. We may chose a set of $|\calP|=ak-|A|$ such petals. We will show later that $|A|=O(k)$ so we may assume $|\calP|=ak-|A|\geq 2(k+1)$ for a sufficiently large constant $a$.
For each $P\in \calP$, define the set:
\[
A_P=A\cup C \cup \bigg( {\textstyle \bigcup_{P'\in \calP\setminus P}} P' \bigg) \ . 
\] 
Let  $\calP'$ contain all $P\in \calP$ such that $c(P)\cap c(A_P)=\emptyset$ and $|c(P)|=|P|$. Suppose $c(C)\cap c(A)=\emptyset$ and $|c(C)|=|C|$. Then the family $\calF= \{c(P\cup C)\}_{P\in \calP'}$ satisfies Property 2. 

To show $\calF$ also satisfies Property 3 consider edges $C\cup Q_1$ and $C\cup Q_2$ such that $c( C\cup Q_1)=c(C\cup P_1)$ and $c( C\cup Q_2)=c(C\cup P_2)$. 
Then $c(Q_1)=c(P_1)$ and $c(Q_2)=c(P_2)$ because $|c(C)|=|C|$ and  $|c(P_1)|=|P_1|$, and $|c(P_2)|=|P_2|$. But $c(P_1)\cap c(P_2)=\emptyset$ implies $c(Q_1)\cap c(Q_2)=\emptyset$ and so $Q_1\cap Q_2 =\emptyset$.

\paragraph{Size of family $\calF$.} It will suffice to show that $c(C)\cap c(A)=\emptyset$ and $|c(C)|=|C|$ with probability $3/4$ and $|\calP'|\geq (k+1)$ with probability $3/4$.  Then $\calF$ satisfies all three properties and has size $(k+1)$ with probability $1/2$.
Suppose $b=\max(4d(|A|+d),8d (|A_P|+d))$. Then,
\[
\prob{c(C)\cap c(A)=\emptyset, |c(C)|=|C|}\geq 1-(d|A|+d^2)/b\geq 3/4 \ .
\]
For each $P\in \calP$, let $X_P=1$ if $c(P)\cap c(A_P)\neq \emptyset$ or $|c(P)|\neq  |P|$ and $X_P=0$ otherwise. 
Then $\expec{\sum X_P}\leq |\calP| (d |A_P|+d^2)/b \leq  |\calP|/8$. By applying the Markov inequality, $\prob{\sum X_P\geq |\calP|/2}\leq 1/4$. Hence, $|\calP'|=|\calP|-\sum X_P \geq |\calP|/2=d(k+1)$ with probability at least $3/4$.

It remains to show that $b=O(k)$ where we omit dependencies on $d$. To do this, it suffices to show $|A|=O(k)$ and $|A_P|=O(k)$. By appealing to Lemma \ref{small_hs},
\[
|A| \leq 
|A_P|  \leq |A|+|C|+d|\calP|\leq  |W|+\sum_{C'\subset C} \hs(M_{C,C'}) +|C|+d|\calP|\leq k+ 2^d dk+d+dak=O(k) \ .
\]
\end{proof}

\begin{lemma}\label{lem_f}
$\prob{F\subseteq E'}\geq 1-1/\poly(k)$. 
\end{lemma}
\begin{proof}
Pick an arbitrary edge $D \in F$.  Consider $H\sim \sample_{b,d,1}$ and let $c:[n]\rightarrow [b]$ be the coloring that defined $H$. We need to show that there is a unique edge that is $c(D)$-colored since then $D$ is an edge in $H$.
It suffices to show that this is the case with probability at least $1/2$ because repeating the process $O(\log k)$ times will ensure that such a family exists with high probability. The result then follows by taking the union bound over all $D\in F$ since $|F| = O(k^{d})$.

Let $S=c(D)$. 
We first define a set $A$ of vertices such that the only edge that is disjoint from $A$ is $D$. Then it follows that $D$ is the unique $S$-colored edge  if $S\cap c(A)=\emptyset$; every other edge intersects with $A$ and hence must share a color with $A$. We define $A$ as follows:
\[A= ( W \setminus D) \cup \big( {\textstyle \bigcup_{C \subset D}} \hs(M_{C,D}) \big) \]
where $W$ is a minimum hitting set and, by a slight abuse of notation, we use $\hs(M_{C,D})$ to denote a minimum hitting set of $M_{C,D}$. Note that $\hs(M_{C,D})$ does not include any vertices in $D$. 
If an edge is disjoint from $(W\setminus D)$ then it must intersect $D$ since $W$ is a hitting set.
Suppose there  exists an edge such that $D\cap D'= C\subset D$ then $D'$ intersects $\hs(M_{C,D})$. 
Hence, the only edge that is disjoint from $A$ includes the vertices in $D$ and hence is equal to $D$ on the assumption that all edges have the same number of vertices.

It remains to show that $S\cap c(A)=\emptyset$ with probability at least $1/2$. If $b\geq 2d|A|$ then we have
\[
\prob{S\cap c(A)=\emptyset}\geq 1- d|A|/b \geq 1/2 \ .
\]
Finally, note that $2d|A|=O(k)$ since
$
|A|\leq |W|+\sum_{C\subset D} \hs(M_{C,D})
\leq k+2^d akd=O(k)
$
by appealing to Lemma \ref{small_hs} and using the fact that $s_G(C)< ak$ for all $C\subset D$ since $D\in F$.
\end{proof}

A result for hypergraph matching follows along similar lines.
\begin{theorem}\label{thm:hmatch}
Suppose $\matching(G) \leq k' = k/d$. With probability $1-1/\poly(k)$, 
$\matching(G')=\matching(G)$.
\end{theorem}

%\section{A General Family of Problems}
%\label{sec:general}
\section{Sampling Kernels for Subgraph Search Problems}
\label{sec:general}

Finally, we consider a class of problems where the objective is to
search for a subgraph $H$ of $G(V,E)$ which satisfies some property $\mathcal{P}$.
In the parametrized setting, we typically search for the largest $H$
which satisfies this property, subject to the promise that the size of
any $H$ satisfying $\mathcal{P}$ is at most $k$.
For concreteness, we assume the size is captured by the number of
vertices in $H$, and our objective is to find a maximum cardinality
satisfying subgraph. 
The  sampling primitive $\sample_{b,2,1}$ can be used here
when $\mathcal{P}$ is preserved under
vertex contraction:
if $G'$ is a vertex contraction of $G$, then any subgraph $H$ of $G'$
satisfying $\mathcal{P}$ also satisfies $\mathcal{P}$ for $G$ (with
vertices suitably remapped). 
Here, the vertex contraction of vertices $u$ and $v$ creates a new vertex
whose neighbors are $\Gamma(u) \cup \Gamma(v)$.
Many well-studied problems posess the required structure, including:
%These include:

\begin{trivlist}
\item
 --- $b$-matching, to find a (maximum cardinality) subgraph $H$ of $G$
  such that the degree of each vertex in $H$ is at most $b$.  Hence, the
  standard notion of matching in Section~\ref{sec:basic:sample} is
  equivalent to 1-matching.
\item
 --- $k$-colorable subgraph, to find a subgraph $H$ that is
  $k$-colorable. The maximum cardinality 2-colorable subgraph forms
  a max-cut, and more generally the maximum cardinality
  $k$-colorable subgraph is a max $k$-cut. 
\item
---  other maximum subgraph problems, such as to find the largest
  subgraph that is a forest, has at least $c$ connected components, or
  is a collection of vertex disjoint paths. 
\end{trivlist}

\begin{theorem}
  Let $\mathcal{P}$ be a graph property preserved under vertex contraction.
  Suppose that the number of vertices in some optimum solution
  $\operatorname{opt}(G)$ is at most $k$. 
  Let $G' \sim \sample_{4k^2,2,1}(G)$.
  With constant probability, we can compute a solution $H$ for
  $\mathcal{P}$ from $G'$ that achieves $|H|=|\opt(G)|$. 
\label{thm:contract}
\end{theorem}

\begin{proof}
We construct a {\em contracted graph} $G''$ from $G'$ based
on the color classes used in the $\sample$ operator:
we contract all vertices that are  assigned the same color by the hash function $c()$. 
Fix an optimum solution $\opt(G)$ with at most $k$ vertices. 
Lemma \ref{lem:markov:ineq} shows that for $b=4k^2$, all vertices involved in $opt(G)$
are hashed into distinct color values.
Hence, the subgraph $\opt(G)$ is a subgraph of $G''$:
for any edge $e=(u,v)\in \opt(G)$, the edge itself was sampled from
the data structure, or else a different edge with the same color
values was sampled, and so can be used interchangeably in $G''$.
Hence, (the remapped form of) $\opt(G)$ persists in $G''$. 
By the vertex contraction property of $\mathcal{P}$, this means that a
maximum cardinality solution for $\mathcal{P}$ in $G''$ is a maximum
cardinality solution in $G$.

Note that for this application of the subgraph sampling primitive, it suffices to implement
the sampling data structure with a counter for each pair of colors:
any non-zero count corresponds to an edge in $G''$. 
\end{proof}

We can follow the same template laid out in Section~\ref{sec:exact} to
generalize to the weighted case (e.g., where the objective is to find
the subgraph satisfying $\mathcal{P}$ with the greatest total
weight).
We can perform the sampling in parallel for each distinct weight
value, and then round each edge weight  to the closest power of
$(1+\epsilon)$ to reduce the number of weight classes to
$O(\epsilon^{-1} \log W)$, with a loss factor of $(1+\epsilon)$.

\newcommand{\Proc}{Proceedings of the~}
\newcommand{\STOC}{Annual ACM Symposium on Theory of Computing (STOC)}
\newcommand{\FOCS}{IEEE Symposium on Foundations of Computer Science (FOCS)}
\newcommand{\SODA}{Annual ACM-SIAM Symposium on Discrete Algorithms (SODA)}
\newcommand{\SOCG}{Annual Symposium on Computational Geometry (SoCG)}
\newcommand{\ICALP}{Annual International Colloquium on Automata, Languages and Programming (ICALP)}
\newcommand{\ESA}{Annual European Symposium on Algorithms (ESA)}
\newcommand{\CCC}{Annual IEEE Conference on Computational Complexity (CCC)}
\newcommand{\RANDOM}{International Workshop on Randomization and Approximation Techniques in Computer Science (RANDOM)}
\newcommand{\APPROX}{International Workshop on  Approximation Algorithms for Combinatorial Optimization Problems  (APPROX)}
\newcommand{\PODS}{ACM SIGMOD Symposium on Principles of Database Systems (PODS)}
\newcommand{\SSDBM}{ International Conference on Scientific and Statistical Database Management (SSDBM)}
\newcommand{\ALENEX}{Workshop on Algorithm Engineering and Experiments (ALENEX)}
\newcommand{\BEATCS}{Bulletin of the European Association for Theoretical Computer Science (BEATCS)}
\newcommand{\CCCG}{Canadian Conference on Computational Geometry (CCCG)}
\newcommand{\CIAC}{Italian Conference on Algorithms and Complexity (CIAC)}
\newcommand{\COCOON}{Annual International Computing Combinatorics Conference (COCOON)}
\newcommand{\COLT}{Annual Conference on Learning Theory (COLT)}
\newcommand{\COMPGEOM}{Annual ACM Symposium on Computational Geometry}
\newcommand{\DCGEOM}{Discrete \& Computational Geometry}
\newcommand{\DISC}{International Symposium on Distributed Computing (DISC)}
\newcommand{\ECCC}{Electronic Colloquium on Computational Complexity (ECCC)}
\newcommand{\FSTTCS}{Foundations of Software Technology and Theoretical Computer Science (FSTTCS)}
\newcommand{\ICCCN}{IEEE International Conference on Computer Communications and Networks (ICCCN)}
\newcommand{\ICDCS}{International Conference on Distributed Computing Systems (ICDCS)}
\newcommand{\VLDB}{ International Conference on Very Large Data Bases (VLDB)}
\newcommand{\IJCGA}{International Journal of Computational Geometry and Applications}
\newcommand{\INFOCOM}{IEEE INFOCOM}
\newcommand{\IPCO}{International Integer Programming and Combinatorial Optimization Conference (IPCO)}
\newcommand{\ISAAC}{International Symposium on Algorithms and Computation (ISAAC)}
\newcommand{\ISTCS}{Israel Symposium on Theory of Computing and Systems (ISTCS)}
\newcommand{\JACM}{Journal of the ACM}
\newcommand{\LNCS}{Lecture Notes in Computer Science}
\newcommand{\RSA}{Random Structures and Algorithms}
\newcommand{\SPAA}{Annual ACM Symposium on Parallel Algorithms and Architectures (SPAA)}
\newcommand{\STACS}{Annual Symposium on Theoretical Aspects of Computer Science (STACS)}
\newcommand{\SWAT}{Scandinavian Workshop on Algorithm Theory (SWAT)}
\newcommand{\TALG}{ACM Transactions on Algorithms}
\newcommand{\UAI}{Conference on Uncertainty in Artificial Intelligence (UAI)}
\newcommand{\WADS}{Workshop on Algorithms and Data Structures (WADS)}
\newcommand{\SICOMP}{SIAM Journal on Computing}
\newcommand{\JCSS}{Journal of Computer and System Sciences}
\newcommand{\JASIS}{Journal of the American society for information science}
\newcommand{\PMS}{ Philosophical Magazine Series}
\newcommand{\ML}{Machine Learning}
\newcommand{\DCG}{Discrete and Computational Geometry}
\newcommand{\TODS}{ACM Transactions on Database Systems (TODS)}
\newcommand{\PHREV}{Physical Review E}
\newcommand{\NATS}{National Academy of Sciences}
\newcommand{\MPHy}{Reviews of Modern Physics}
\newcommand{\NRG}{Nature Reviews : Genetics}
\newcommand{\BullAMS}{Bulletin (New Series) of the American Mathematical Society}
\newcommand{\AMSM}{The American Mathematical Monthly}
\newcommand{\JAM}{SIAM Journal on Applied Mathematics}
\newcommand{\JDM}{SIAM Journal of  Discrete Math}
\newcommand{\JASM}{Journal of the American Statistical Association}
\newcommand{\AMS}{Annals of Mathematical Statistics}
\newcommand{\JALG}{Journal of Algorithms}
\newcommand{\TIT}{IEEE Transactions on Information Theory}
\newcommand{\CM}{Contemporary Mathematics}
\newcommand{\JC}{Journal of Complexity}
\newcommand{\TSE}{IEEE Transactions on Software Engineering}
\newcommand{\TNDE}{IEEE Transactions on Knowledge and Data Engineering}
\newcommand{\JIC}{Journal Information and Computation}
\newcommand{\ToC}{Theory of Computing}
\newcommand{\MST}{Mathematical Systems Theory}
\newcommand{\Com}{Combinatorica}
\newcommand{\NC}{Neural Computation}
\newcommand{\TAP}{The Annals of Probability}
\newcommand{\TCS}{Theoretical Computer Science}
\newcommand{\IPL}{Information Processing Letter}
\newcommand{\Algorithmica}{Algorithmica}

\newpage

{
\bibliographystyle{abbrv} \bibliography{dynamic}
}

\newpage

\appendix

\section{Omitted Proofs}
\label{sec:omitted}
\begin{proof}[Proof of Lemma~\ref{lem:markov:ineq}]
Let $b=4\epsilon r$. For a vertex $u\in U$, let $I_u$ be the indicator random variable that equals one if there exists $u'\in U\setminus \{u\}$ such that $c(u)=c(u')$.
Since $c$ is pairwise independent, 
$$
\prob{I_u}\leq \sum_{u'\in U\setminus \{u\}} \prob{c(u)=c(u')}
= \sum_{u'\in U\setminus \{u\}} 1/b < r/b=\epsilon/4 
\ .
$$
Let $I=\sum_{u\in U} I_u$ and note that $\expec{I}\le \epsilon r/4$. Then Markov's inequality implies $\prob{I\ge \epsilon r}\leq 1/4$.
\end{proof}

\begin{proof}[Proof of Theorem~\ref{thm:sample:distribution}]
To sample a graph from $\sample_{b,d,r}$ we simply sample $r$ graphs from $\sample_{b,d,r}$ in parallel.
To draw a sample from $\sample_{b,d,1}$, we employ one instance of an
$\ell_0$-sampling primitive for each of the $O(b^d)$ edge colorings~\cite{JST11,CormodeF14}: Given a dynamic graph stream, the $\ell_0$-sampler returns FAIL with probability at most $\delta$.
Otherwise, it returns an edge chosen uniformly at random amongst the edges that have been inserted and not deleted. If there are no such edges, the $\ell_0$-sampler returns NULL. The $\ell_0$-sampling primitive can be implemented using $O(\log^2 n\log\delta^{-1})$ bits of space and $O(\polylog n)$ update time.
In some cases, we can make use of simpler deterministic data structures.
For Theorem~\ref{thm:exact}, we can replace the $\ell_0$ sampler with a counter
and the exclusive-or of all the edge identifiers, since we only
require to recover edges when they are unique within their color
class.
For Theorem~\ref{thm:contract}, we only require a counter.
In both cases, the space cost is reduced to $O(\log n$). 

At the start of the stream we choose a pairwise independent hash function $c:V\rightarrow [b]$. For each weight $w$ and subset $S\subset [b]$ of size $d$, this hash function defines a sub-stream corresponding to the $S$-colored edges of weight $w$. We then use $\ell_0$-sampling on each sub-stream to select a random edge from $E_S$.
\end{proof}

\begin{proof}[Proof of Lemma \ref{eq1suffices}]
We first argue that $\vc(G')=\vc(G)$.  Since the vertex cover of $G$ is of size at most $2k$, we know every vertex in $U$ must be in the vertex cover of both $G$ and $G'$ since the degrees of such vertices in both graphs are strictly greater than $2k$. This follows because if a vertex in $U$ was not in the minimum vertex cover then all its neighbors need to be in the vertex cover.

We next argue that $\matching(G')=\matching(G)$. If property~\eqref{eq:magic}
is satisfied then $G'$ contains a matching of size
$\matching(F)+\sou \geq \matching(G)$ since we may choose the optimum
matching in $F$ and then still be able to match every vertex in
$U$. This follows because the optimum matching in $F$ ``consumes'' at
most $2k$ potential endpoints, since $\matching(G)\leq k$. Hence, each of the (at most $2k$) vertices in $U$ can still be matched to $3k$ possible vertices. 
\end{proof}

\begin{proof}[Proof of Lemma~\ref{lem:removingwe}]
Let $E\setminus E'=\{e_1,e_2, \ldots e_t\}$ and let $G'_i$ be the graph formed by removing $\{e_1,\ldots, e_i\}$ from $G$. So $G'_0=G$ and $G'_t=G'$. For the sake of contradiction, suppose  $\matching(G)> \matching(G')$ and let $r$ be the minimal value such that $\matching(G)> \matching(G'_r)$. 

By the minimality of $r$,  $\matching(G)= \matching(G'_{r-1})$. Consider the maximum weight matching $M$ in $G'_{r-1}$. If $e_r\not \in M$ then $\matching(G)=\matching(G'_{r-1})= \matching(G'_r)$ and we have a contradiction. If $e_r\in M$, let $u,v$ be the endpoints of $e_r$ and the weight of $e_r$ be $w$. Without loss of generality $\deg_{G'_{r}}^w(u)\geq d_{G'}^w(u)\geq 5k$. Hence, there exists edge $ux$ of weight $w$ in $G'_{r}$ where $x$ is not an endpoint in $M$. Therefore, the matching $\left (M\setminus \{e_r\}\right )\cup \{ux\}$ is contained in $G'_r$ and has the same weight as $M$. Hence, $\matching(G)=\matching(G'_{r-1})= \matching(G'_r)$ and we again have a contradiction. \end{proof}

\begin{proof}[Proof of Corollary \ref{cor:semimatch}]
For $1\leq i\leq \log n$, let $G'_i\sim \sample_{b,1,r}$ where $r=O(2^i \alpha^{-2}\log k)$ and $b=2^{i+1}/\alpha$. These graphs can be generated in $\tilde{O}(n^2 \alpha^{-3})$ space. For some $i$, $2^i\leq \match(G)<2^{i+1}$ and hence $\match(G'_i)=\Omega(\match(G)/\alpha)$.
\end{proof}

\begin{proof}[Proof of Lemma~\ref{lemma:smallf}]
For the sake of contradiction assume $|F| > d!(ak)^d$. Then, by the Sunflower Lemma, $F$ contains a $(ak+1)$-sunflower. If the core of this sunflower is empty, $F$ has a matching of size $(ak+1)$ and therefore cannot have a hitting set of size at most $k$. If the sunflower has a non-empty core $C$, then some edge $D \in F$ contains $C$, which contradicts the definition of $F$. Therefore, $|F| \leq d!(ak)^d$.
% = O(k^d)$.
%\end{proof}
%
%\begin{proof}[Proof of Lemma~\ref{lemma:smallu}]

To prove $|U'| \leq  (d-1)!k^{d-1}$,  first note that $|C'| \leq d-1$ for all $C' \in U'$.  For the sake of contradiction assume that $|U'| > (d-1)!k^{d-1}$. Then, by the Sunflower Lemma again, $U'$ contains a $(k+1)$-sunflower. Note that it is a sunflower of cores, not hypergraph edges. Let $C_1, C_2,\ldots , C_{k+1}$ be the sets in the sunflower. Each of these sets has to contain at least one vertex of the minimum hitting set. Therefore, if $C_1, C_2,\ldots , C_{k+1}$ are disjoint (i.e., the core of the sunflower is empty), $U'$ has a matching of size $(k+1)$ and cannot have a hitting set of size at most $k$. If the sunflower has a non-empty core $C^*$, we will show that union of the maximum sunflowers with cores $C_1, C_2,\ldots , C_{k+1}$ contains a sunflower with $k+1$ edges with core $C^*\subset C_1\in U'$. This contradicts  the definition of $U'$ and therefore $|U'| \leq (d-1)!k^{d-1} = O(k^{d-1})$. To construct the sunflower on $C^*$, for $i=1, \ldots, k+1$, we  pick an edge $D_i$ in the maximum sunflower with core $C_i$ such that $D_i\cap C_j=C^*$ for $j\neq i$ and $D_i\cap D_j=C^*$ for $j<i$. This is possible if $a$ is sufficiently large.
\end{proof}

\begin{proof}[Proof of Lemma \ref{small_hs}]
Consider the size of minimum hitting set of $M_{C,D}$. If $\hs(M_{C,D}) > s_G(C) d$, then $M_{C,D}$ has a matching of size greater than $s_G(C)$. This matching together with the set $C$ forms a sunflower with core $C$ and over $s_G(C)$ petals, which contradicts the assumption. Therefore, $\hs(M_{C,D}) \leq s_G(C) d$ as claimed.
\end{proof}

\begin{proof}[Proof of Theorem \ref{thm:hmatch}]
$\hs(G) \leq dk' = k$. Let $M$ be the matching. $F \cap M$ is preserved in $G'$. Consider an edge $D \in M$ such that $C \subseteq D$ for some $C \in U$. Then in $G'$ we can find (by Lemma \ref{lem_u}) at least $k+1$ petals in a sunflower with core either $C$ itself or some $C' \subset C$. At most $k$ of those intersect $M \setminus \{D\}$. Therefore, there is still at least one edge we can pick for the matching.
\end{proof}

\section{Matchings in Planar and Bounded-Arboricity Graphs } \label{sec:planar}
In this section, we present an algorithm for estimating the size of the matching in a graph of bounded arboricity.
Recall that a graph has \emph{arboricity} $\nu$ if its edges can be partitioned into at most $\nu$ forests. 
In particular, it can be shown that a planar graph has arboricity at most 3. We will make repeated use of the fact that the average degree of every subgraph of a graph with arboricity $\nu$ is at most $2\nu$. 

Our algorithm is based on an insertion-only streaming algorithm due to Esfandiari et al.~\cite{EsfandiariHLMO15}. 
They first proved upper and lower bounds on the size of the maximum matching in a graph of arboricity $\nu$. 

\begin{lemma}[Esfandiari et al.~\cite{EsfandiariHLMO15}]
\label{lem:mainproperty}
For any graph $G$ with arboricity $\nu$, define a vertex to be \emph{heavy} if its degree is at least $2\nu+3$ and define an edge to be \emph{shallow} if it is not incident to a heavy vertex. Then,
\[\frac{\max\{ h, s\}}{2.5\nu +4.5} \leq \match(G) \leq 2\max\{h,s\} \ .\] 
where $h$ is the number of heavy  vertices and $s$ is the number of shallow edges.
\end{lemma}
To estimate $\max \{h,s\}$,  Esfandiari et al.~sampled a set of
vertices $Z$ and (a) computed the exact degree of these vertices, then (b) found the set of all edges in the induced subgraph on these vertices. The fraction of heavy vertices in $Z$ and shallow edges in the induced graph are then used to estimate $h$ and $s$. By choosing the size of $Z$ appropriately, they showed that the resulting estimate was sufficiently accurate on the assumption that $\max \{h,s\}$ is large. In the case where $\max \{h,s\}$ is small, the maximum matching is also small and hence a maximal matching could be constructed in small space using a greedy algorithm.

\paragraph{Algorithm for Dynamic Graph Streams.} In the dynamic graph
stream model, it is not possible to construct a maximal
matching. However, we may instead use the algorithm of
Theorem \ref{thm:exact} to find the exact size of the maximum
matching. Furthermore we can still recover the induced subgraph on
sampled vertices $Z$ via a sparse recovery sketch \cite{GilbertI10}. This can be done space-efficiently because the number of edges is at most $2\nu |Z|$. Lastly, rather than fixing the size of $Z$, we consider sampling each vertex independently with a fixed probability as this simplifies the analysis significantly. The resulting algorithm is as follows:

\begin{enumerate}
\item Invoke algorithm of Theorem \ref{thm:exact} for $k=2 n^{2/5}$ and let $r$ be the reported matching size. 
\item In parallel, sample vertices with probability $p=8\epsilon^{-2}n^{-1/5}$  and let $Z$ be the set of sampled vertices. 
Compute the degrees of vertices in $Z$ and maintain a $2\nu |Z|$-sparse recovery 
sketch of the edges in the induced graph on $Z$. Let $s_Z$ be the number of shallow edges in the induced graph on $Z$ and let $s_Z$ be the number of heavy vertices in $Z$. Return $\max\{r,h_Z/p,s_Z/p^2\}$.
\end{enumerate}

\paragraph{Analysis.}
Our analysis relies on the following lemma that shows that $\max\{h_Z/p,s_Z/p^2\}$ is a $1+\epsilon$ approximation for $\max \{s,h\}$ on the assumption that $\max \{s,h\}\geq n^{2/5}$.

 \begin{lemma}
 \label{lem:shallow:edge} $   \prob {|\max\{h_Z/p,s_Z/p^2\}-\max \{ s,h\}|\le \epsilon\cdot \max\{n^{2/5}, s,h\}} \ge 4/5 \ . $
 \end{lemma}
\begin{proof}
First we show $s_Z/p^2$ is a sufficiently good estimate for $s$.
Let $S$ be the set of shallow edges in $G$ and let $E_Z$ be the set of edges in the induced graph on $Z$.
For each shallow edge $e\in S$, define an indicator random variable $X_e$ where $X_e=1$ iff $e\in E_Z$ and note that $s_Z=\sum_{e\in S}X_e$. Then,
\[
\expec{s_Z}=sp^2 \quad \mbox{ and }  \quad
\var{s_Z}=
%\sum_{e\in S} \var{X_e}+\sum_{e\neq e'} \covar{X_{e},X_{e'}}
%=sp^2(1-p^2)+\sum_{e\neq e'} \covar{X_{e},X_{e'}} \ .
\sum_{e\in S} \sum_{e'\in S} \expec{X_{e}X_{e'}}-\expec{X_{e}}\expec{X_{e'}} \ .
\] 
Note that 
\[
\sum_{e'\in S} \expec{X_{e}X_{e'}}-\expec{X_{e}}\expec{X_{e'}}
=
\begin{cases}
p^2-p^4 & \mbox{if $e=e'$} \\
p^3-p^4 & \mbox{if $e$ and $e'$ share exactly one endpoint} \\
0 & \mbox{if $e$ and $e'$ share no endpoints} 
\end{cases} \ .
\]
and since there are at most $2\nu+3$ edges that share an endpoint with a shallow edge,
\[
\var{s_Z}\leq s(p^2-p^4+(2\nu+3)p^3-p^4)\leq 2sp^2
\]
on the assumption that $(2\nu+3)\leq 1/p$.
We then use  Chebyshev's inequality  to obtain 

\begin{equation}\label{szeq}
%\prob {|s_Z/p^2-s|\le \epsilon\cdot \max\{n^{2/5}, s\}} 
%&=&
\prob {|s_Z-sp^2 |\le p^2 \epsilon\cdot \max\{n^{2/5}, s\}} 
 \leq  
\frac{2sp^2}{(p^2 \epsilon\cdot \max\{n^{2/5}, s\})^2} 
%\\
%& =& 
%\frac{2}{p^2 \epsilon^2 \cdot \max\{n^{4/5}/{s}, s\})} 
\leq 9/10 \ .
\end{equation}

Next we show that $h_Z/p$ is a sufficiently good estimate for $h$.
 Let $H$ denote the set of $h$ heavy vertices in $G$ and define an indicator random variable $Y_v$  for each $v\in H$, where $Y_v=1$ iff $v\in Z$. Note that $h_Z=\sum_{v\in H} Y_v$ and $\expec{h_Z}=hp$. Then, by an application of the Chernoff-Hoeffding bound,
 \begin{equation}\label{hzeq}
% \prob{|r_2-h|\geq \epsilon  \max\{h,n^{2/5}\}}= 
 \prob{|h_Z-hp|\geq \epsilon  p  \max\{h,n^{2/5}\} }\leq \exp(-\epsilon^2  pn^{2/5}/3) \leq 9/10 \ . 
 \end{equation}
 
 Therefore, it follows from Eq.~\ref{szeq} and \ref{hzeq} that
 $\prob{\max \{h_Z/p,s/p^2 \} \leq \epsilon \max\{h,s,n^{2/5}\}}\geq 4/5$.
%  \ .\]
%  
\end{proof}

%The resulting theorem is proved in the appendix.

\begin{theorem}
\label{thm:esimate:planar:matching}
There exists a $\tilde{O}(\nu \epsilon^{-2}n^{4/5} \log \delta^{-1} )$-space dynamic graph stream algorithm that returns a $(5\nu+9)(1+\epsilon)^2$ approximation of $\match(G)$ with probability at least $1-\delta$ where $\nu$ is the arboricity of $G$. 
\end{theorem}
\begin{proof}%[Proof of Theorem \ref{thm:esimate:planar:matching}]
To argue the approximation factor, first suppose $\match(G)\leq 2n^{2/5}$. In this case $r=\match(G)$ and $\max \{s,h\}\leq (2.5\nu+4.5) \match(G)$ by appealing to Lemma \ref{lem:mainproperty}. Hence,
\[\match(G) \leq \max\{r,h_Z/p,s_Z/p^2\}\leq (2.5\nu+4.5) \match(G)\]
Next suppose  $\match(G)\geq 2n^{2/5}$. In this case, $\max\{s,h\}\geq n^{2/5}$  by Lemma \ref{lem:mainproperty}. Therefore, by Lemma \ref{lem:shallow:edge}, $\max\{h_Z/p,s_Z/p^2\}=(1\pm \epsilon) \max\{s,h\}$, and so
\[\frac{\match(G)}{2(1+\epsilon)} \leq \max\{r,h_Z/p,s_Z/p^2\}\leq (1+\epsilon) \max\{s,h\} \leq (1+\epsilon) (2.5\nu+4.5) \match(G) \]

To argue the space bound, recall that the algorithm used in Theorem \ref{thm:exact} requires $\tilde{O}(n^{4/5})$ space. Note that $|Z|\leq 2np=\tilde{O}(\epsilon^{-1/2} n^{4/5})$ with high probability. Hence, to sample the vertices $Z$ and maintain a $2\nu |Z|$-sparse recovery requires $\tilde{O}(n^{4/5}\nu)$ space.
\end{proof}

\section{Lower Bounds}
\label{sec:lowerbounds}

\subsection{Matching and Hitting Set Lower Bounds}

The following theorem establishes that the space-use of our matching, vertex cover, hitting set, and hyper matching algorithms is optimal up to logarithmic factors.

\begin{theorem}\label{lower_hs}
Any (randomized) parametrized streaming algorithm for the minimum $d$-hitting set or maximum (hyper)matching problem with parameter $k$ requires $\Omega(k^d)$ space.
\end{theorem}

\begin{proof}
We reduce from the \textsc{Membership} communication problem:

\begin{center}
\noindent\framebox{\begin{minipage}{0.95\columnwidth}
\textsc{Membership}\\
\emph{Input}: Alice has a set $X\subseteq [n]$, and Bob has an element $1\leq x\leq n$.\\
\emph{Question}: Bob wants to check whether $x\in X$.
\end{minipage}}
\end{center}

There is a lower bound of $\Omega(n)$ bits of communication from Alice to Bob, even allowing randomization~\cite{Ablayev96}.

Let $S = s_1s_2...s_n$ be the characteristic string of $X$, i.e. a binary string such that $s_i = 1$  iff $i \in X$. Let $k = \sqrt[d]{n}$. Fix a canonical mapping $h:[n] \rightarrow [k]^d$. This way we can view an $n$ bit string as an adjacency matrix of a $d$-partite graph. Construct the following graph $G$ with $d$ vertex partitions $V_1, V_2,..., V_d$:
\begin{itemize}
\item Each partition $V_i$ has $dk$ vertices: for each $j \in [k]$ create vertices $v^{*}_{i,j}$, $v^{1}_{i,j}$, $v^{2}_{i,j}$,..., $v^{d-1}_{i,j}$.
\item Alice inserts a hyperedge $(v^{*}_{1,j_1}, v^{*}_{2,j_2},..., v^{*}_{d,j_d})$ iff the corresponding bit in the string $S$ is 1, i.e., $s_a = 1$ where $h(a) = (j_1,j_2,...,j_d)$.
\item Let $h(x) = (J_1, J_2,..., J_d)$. Bob inserts edge $(v^{*}_{i,j}, v^{1}_{i,j}, v^{2}_{i,j},..., v^{d-1}_{i,j})$ iff $j \neq J_i$.
\end{itemize}
Alice runs the hitting set algorithm on the edges she is inserting using space $f(k)$. Then she sends the memory contents of the algorithm to Bob, who finishes running the algorithm on his edges.

The minimum hitting set should include vertices $v^{*}_{i,j}$ such that $j \neq J_i$. If edge $(v^{*}_{1,J_1}, v^{*}_{2,J_2},..., v^{*}_{d,J_d})$ is in the graph, we also need to include one of its vertices. Therefore,
$$x\in X \iff s_{x} = 1 \iff (v^{*}_{1,J_1}, v^{*}_{2,J_2},..., v^{*}_{d,J_d}) \textrm{ is in } G \iff \hs(G) = dk-d+1$$
On the other hand,
$$x \not\in X \iff s_{x} = 0 \iff (v^{*}_{1,J_1}, v^{*}_{2,J_2},..., v^{*}_{d,J_d}) \textrm{ is not in } G \iff \hs(G) = dk-d$$
Alice only sends $f(k)$ bits to Bob. Therefore, $f(k) = \Omega(n) = \Omega(k^d)$.

For the lower bound on matching we use the same construction. For each vertex $v^{*}_{i,j}$ such that $j \neq J_i$ maximum matching should include $(v^{*}_{i,j}, v^{1}_{i,j}, v^{2}_{i,j},..., v^{d-1}_{i,j})$. If edge $(v^{*}_{1,J_1}, v^{*}_{2,J_2},..., v^{*}_{d,J_d})$ is in the graph, we include it in the matching as well. Therefore,
$$x\in X \iff s_{x} = 1 \iff (v^{*}_{1,J_1}, v^{*}_{2,J_2},..., v^{*}_{d,J_d}) \textrm{ is in } G \iff \match(G) = dk-d+1$$
And
$$x \not\in X \iff s_{x} = 0 \iff (v^{*}_{1,J_1}, v^{*}_{2,J_2},..., v^{*}_{d,J_d}) \textrm{ is not in } G \iff \match(G) = dk-d$$
\end{proof}

\subsection{Lower Bounds for Problems considered by Fafianie and Kratsch~\cite{kratsch}}

\textbf{Comparison with Lower Bounds for Streaming Kernels}: Fafianie and Kratsch~\cite{kratsch} introduced the notion of kernelization in the streaming setting as follows:

\begin{definition}
A 1-pass streaming kernelization algorithm is receives an input $(x,k)$ and returns a kernel, with the restriction that the space usage of the algorithm is bounded by $p(k)\cdot \log |x|$ for some polynomial $p$.
\end{definition}

Fafianie and Kratsch~\cite{kratsch} gave lower bounds for several parameterized problems. In particular, they showed that:

\begin{itemize}
\item Any 1-pass kernel for \textsc{Edge Dominating Set}$(k)$ requires $\Omega(m)$ bits, where $m$ is the number of edges. However, there is a 2-pass kernel which uses $O(k^{3}\cdot \log n)$ bits of local memory and $O(k^2)$ time in each step and returns an equivalent instance of size $O(k^{3}\cdot \log k)$.
\item The lower bound of $\Omega(m)$ bits for any 1-pass kernel also holds for several other problems such as \textsc{Cluster Editing}$(k)$, \textsc{Cluster Deletion}$(k)$, \textsc{Cluster Vertex Deletion}$(k)$, \textsc{Cograph Vertex Deletion}$(k)$, \textsc{Minimum Fill-In}$(k)$, \textsc{Edge Bipartization}$(k)$, \textsc{Feedback Vertex Set}$(k)$, \textsc{Odd Cycle Transversal}$(k)$, \textsc{Triangle Edge Deletion}$(k)$, \textsc{Triangle Vertex
Deletion}$(k)$, \textsc{Triangle Packing}$(k)$, \textsc{$s$-Star Packing}$(k)$, \textsc{Bipartite Colorful Neighborhood}$(k)$.
\item Any $t$-pass kernel for \textsc{Cluster Editing}$(k)$ and \textsc{Minimum Fill-In}$(k)$ requires $\Omega(n/t)$ space.
\end{itemize}

\eat{
\begin{observation}
Any lower bound for parameterized streaming kernels also holds for the parameterized streaming algorithms.
\end{observation}
\begin{proof}
The idea is that if we can actually solve the problem, then we can get a trivial kernel as follows: Fix a trivial YES instance say $I_1$ and a trivial NO instance $I_2$. Now, solve the given instance using the parameterized algorithm and according to the answer then map it predetermined YES or NO instance. This shows that given a parameterized algorithm one can achieve a streaming kernel with the same space requirements.
\end{proof}
}

In this section, we give $\Omega(n)$ lower bounds for the space complexity of all the problems considered by Fafianie and Kratsch. In addition, we also consider some other problems such as \textsc{Path}$(k)$ which were not considered by Fafianie and Kratsch. A simple observation shows that any lower bound for parameterized streaming kernels also transfers for the parameterized streaming algorithms. Thus the results of Fafiane and Kratsch~\cite{kratsch} also give lower bounds for the parameterized streaming algorithms for these problems. However, our lower bounds have the following advantage over the results of~\cite{kratsch}:

\begin{itemize}
\item All our lower bounds also hold for \emph{randomized algorithms}, whereas the kernel lower bounds were for deterministic algorithms.
\item With the exception of \textsc{Edge Dominating Set}$(k)$, all our lower bounds also hold for \emph{constant number of passes}.
\end{itemize}

\subsubsection{Lower Bound for \textsc{Edge Dominating Set}}
\label{sec:lb-eds}

We now show a lower bound for the \textsc{Edge Dominating Set}$(k)$ problem.
\begin{definition}
Given a graph $G=(V,E)$ we say that a set of edges $X\subseteq E$ is an edge dominating set if every edge in $E\setminus X$ is incident on some edge of $X$.
\end{definition}

\begin{center}
\noindent\framebox{\begin{minipage}{0.95\columnwidth}
\textsc{Edge Dominating Set}$(k)$ \hfill \emph{Parameter}: $k$\\
\emph{Input}: An undirected graphs $G$ and an integer $k$\\
\emph{Question}: Does there exist an edge dominating set $X\subseteq E$ of size at most $k$?
\end{minipage}}
\end{center}

\eat{
We reduce from the \textsc{Membership} problem:

\begin{center}
\noindent\framebox{\begin{minipage}{0.95\columnwidth}
\textsc{Membership}\\
\emph{Input}: Alice has a set $X\subseteq [n]$, and Bob has an element $1\leq x\leq n$.\\
\emph{Question}: Bob wants to check whether $x\in X$.
\end{minipage}}
\end{center}
}

\begin{theorem}
For the \textsc{Edge Dominating Set}$(k)$ problem, any (randomized) streaming algorithm needs $\Omega(n)$ space
\label{thm:eds:lower}.
\end{theorem}
\begin{proof}
Given an instance of {\sc Membership}, we create a graph $G$ on $n+2$ vertices as follows. For each $i\in [n]$ we create a vertex $v_i$. Also add two special vertices $a$ and $b$. For every $y\in X$, add the edge $(a, y)$. Finally add the edge $(b,x)$.

Now we will show that $G$ has an edge dominating set of size 1 iff {\sc Membership} answers YES. In the first direction suppose that $G$ has an edge dominating set of size 1. Then it must be the case that $x\in X$: otherwise for a minimum edge dominating set we need one extra edge to dominate the star incident on $a$, in addition to the edge $(b,x)$ dominating itself. Hence {\sc Membership} answers YES. In reverse direction, suppose that {\sc Membership} answers YES. Then the edge $(a,x)$ is clearly an edge dominating set of size 1.

Therefore, any (randomized) streaming algorithm that can determine whether a graph has an edge dominating set of size at most $k=1$ gives a communication protocol for {\sc Membership}, and hence requires $\Omega(n)$ space.
\end{proof}

\subsubsection{Lower Bound for \textsc{$\mathcal{G}$-Free Deletion}}
\label{sec:lb-deletion}

\begin{definition}
A set of connected graphs $\mathcal{G}$ is \emph{bad} if there is a minimal (under operation of taking subgraphs) graph $H\in \mathcal{G}$ such that $2P_{2}\subseteq H$, where $P_2$ is a path on 2 vertices.
\end{definition}

For any bad set of graphs $\mathcal{G}$, we now show a lower bound for the following general problem:

\begin{center}
\noindent\framebox{\begin{minipage}{0.95\columnwidth}
\textsc{$\mathcal{G}$-Free Deletion}$(k)$ \hfill \emph{Parameter}: $k$\\
\emph{Input}: A bad set of graphs $\mathcal{G}$, an undirected graph $G=(V,E)$ and an integer $k$\\
\emph{Question}: Does there exist a set $X\subseteq V$ such that $G\setminus X$ contains no graph from $\mathcal{G}$?
\end{minipage}}
\end{center}

The reduction from the {\sc Disjointness} problem in communication complexity.

\begin{center}
\noindent\framebox{\begin{minipage}{0.95\columnwidth}
\textsc{Disjointness}\\
\emph{Input}: Alice has a string $x\in \{0,1\}^n$ given by
$x_{1}x_{2}\ldots x_n$. Bob has a string $y \in \{0, 1\}^n$ given by $y_{1}y_{2}\ldots y_n$.\\
\emph{Question}: Bob wants to check if $\exists\ i\in [n]$ such that $x_{i}=y_i=1$.
\end{minipage}}
\end{center}

There is a lower bound of $\Omega(n/p)$ bits of communication between
Alice and Bob, even allowing $p$-rounds and randomization~\cite{nisan}.

\begin{theorem}
For a \emph{bad} set of graphs $\mathcal{G}$, any $p$-pass (randomized) streaming algorithm for the $\textsc{$\mathcal{G}$-Free Deletion}$ problem needs $\Omega(n/p)$ space
\label{thm:fvs:lower}.
\end{theorem}
\begin{proof}

Given an instance of {\sc Disjointness}, we create a graph $G$ which
consists of $n$ disjoint copies say $G_1, G_2, \ldots, G_n$ of
$H':=H\setminus 2P_2$. Let the two edges removed from $H$ to get $H'$
be $e_1$ and $e_2$.  For each $i\in [n]$, to the copy $G_i$ of $H'$ we
add the edge $e_1$ iff $x_i=1$ and the edge $e_2$ iff $y_i=1$. We now
show that the resulting graph $G$ contains a copy of $H$ if and only
if {\sc Disjointness} answers YES.

Suppose that {\sc Disjointness} answers YES. So there is a $j\in [n]$ such that $x_j = 1= y_j$. Therefore, to the copy $G_j$ of $H'$ we would have added the edges $e_1$ and $e_2$ which would complete it into $H$. So $G$ contains a copy of $H$. In other direction, suppose that $G$ contains a copy of $H$. Note that since we add $n$ disjoint copies of $H'$ and add at most two edges ($e_1$ and $e_2$) to each copy, it follows that each connected component of $G$ is in fact a subgraph of $H=H'\cup(e_1 + e_2)$. Since $H$ is connected and $G$ contains a copy of $H$, some connected component of $G$ must exactly be the graph $H$, i.e, to some copy $G_i$ of $H'$ we must have added both the edges $e_1$ and $e_2$. This implies $x_i = 1 = y_i$, and so {\sc Disjointness} answers YES.

Since each connected component of $G$ is a subgraph of $H$, the minimality of $H$ implies that $\mathcal{G}$ contains a graph from $G$ iff $G$ contains a copy of $H$, which in turn is true iff {\sc Disjointness} answers YES. Therefore, any $p$-pass (randomized) streaming algorithm that can determine whether a graph is $\mathcal{G}$-free (i.e., answers the question with $k=0$) gives a communication protocol for {\sc Disjointness}, and hence requires $\Omega(n/p)$ space.
\end{proof}

This implies lower bounds for the following set of problems:
\begin{theorem}
For each of the following problems, any $p$-pass (randomized) algorithm requires $\Omega(n/p)$ space: \textsc{Feedback Vertex Set}$(k)$, \textsc{Odd Cycle Transversal}$(k)$, \textsc{Even Cycle Transversal}$(k)$ and \textsc{Triangle Deletion}$(k)$.
\end{theorem}
\begin{proof}
We first define the problems below:

\begin{center}
\noindent\framebox{\begin{minipage}{0.95\columnwidth}
\textsc{Feedback Vertex Set}$(k)$ \hfill \emph{Parameter}: $k$\\
\emph{Input}: An undirected graph $G=(V,E)$ and an integer $k$\\
\emph{Question}: Does there exist a set $X\subseteq V$ of size at most $k$ such that $G\setminus X$ has no cycles?
%~\\
\end{minipage}}
\end{center}

\begin{center}
\noindent\framebox{\begin{minipage}{0.95\columnwidth}
\textsc{Odd Cycle Transversal}$(k)$ \hfill \emph{Parameter}: $k$\\
\emph{Input}: An undirected graph $G=(V,E)$ and an integer $k$\\
\emph{Question}: Does there exist a set $X\subseteq V$ of size at most $k$ such that $G\setminus X$ has no odd cycles?
%~\\
\end{minipage}}
\end{center}
\begin{center}
\noindent\framebox{\begin{minipage}{0.95\columnwidth}
\textsc{Even Cycle Transversal}$(k)$ \hfill \emph{Parameter}: $k$\\
\emph{Input}: An undirected graph $G=(V,E)$ and an integer $k$\\
\emph{Question}: Does there exist a set $X\subseteq V$ of size at most $k$ such that $G\setminus X$ has no even cycles?
%~\\
\end{minipage}}
\end{center}
\begin{center}
\noindent\framebox{\begin{minipage}{0.95\columnwidth}
\textsc{Triangle Deletion}$(k)$ \hfill \emph{Parameter}: $k$\\
\emph{Input}: An undirected graph $G=(V,E)$ and an integer $k$ \\
\emph{Question}: Does there exist a set $X\subseteq V$ of size at most $k$ such that $G\setminus X$ has no triangles?
\end{minipage}}
\end{center}

Now we show how each of these problems can be viewed as a \textsc{$\mathcal{G}$-Free Deletion} problem for an appropriate choice of \emph{bad} $\mathcal{G}$.
\begin{itemize}
\item \textsc{Feedback Vertex Set}$(k)$: Take $\mathcal{G}=\{C_3, C_4, C_5, \ldots\}$ and $H=C_3$
\item \textsc{Odd Cycle Transversal}$(k)$: Take $\mathcal{G}=\{C_3, C_5, C_7, \ldots\}$ and $H=C_3$
\item \textsc{Even Cycle Transversal}$(k)$: Take $\mathcal{G}=\{C_4, C_6, C_8, \ldots\}$ and $H=C_4$
\item \textsc{Triangle Deletion}$(k)$: Take $\mathcal{G}=\{C_3\}$ and $H=C_3$
%\item Cluster Vertex Deletion: Take $\mathcal{G}=\{\text{All non-cliques}\}$ and $H=P_3$
\end{itemize}
We verify the conditions for \textsc{Feedback Vertex Set}$(k)$; the proofs for other problems are similar. Note that the choice of $\mathcal{G}=\{C_3, C_4, C_5, \ldots\}$ and $H=C_3$ implies that $\mathcal{G}$ is \emph{bad} since each graph in $\mathcal{G}$ is connected, the graph $H$ belongs to $\mathcal{G}$ and is a minimal element of $\mathcal{G}$ (under operation of taking subgraphs). Finally, finding a set $X$ such that the graph $G\setminus X$ is $\mathcal{G}$-free implies that it has no cycles, i.e., $X$ is a feedback vertex set for $G$.
\end{proof}

It is easy to see that the same proofs also work for the edge deletion versions of the \textsc{Odd Cycle Transversal}$(k)$, \textsc{Even Cycle Transversal}$(k)$ and the \textsc{Triangle Deletion}$(k)$ problems.

\subsubsection{\textsc{$\mathcal{G}$-Editing}}
\label{sec:lb-editing}

\begin{definition}
A set of graphs $\mathcal{G}$ is \emph{good} if there is a minimal (under operation of taking subgraphs) connected graph $H\in \mathcal{G}$ such that $2P_{2}\subseteq H$, where $P_2$ is a path on 2 vertices.
\end{definition}

For any \emph{good} set of graphs $\mathcal{G}$, we now show a lower bound for the following general problem:

\begin{center}
\noindent\framebox{\begin{minipage}{0.95\columnwidth}
\textsc{$\mathcal{G}$-Editing}$(k)$ \hfill \emph{Parameter}: $k$\\
\emph{Input}: A graph class $\mathcal{G}$, an undirected graph $G=(V,E)$ and an integer $k$\\
\emph{Question}: Does there exist a set $X$ of $k$ edges such that $(V, E\cup X)$ contains a graph from $\mathcal{G}$?
\end{minipage}}
\end{center}

\begin{theorem}
For a \emph{good} set of graphs $\mathcal{G}$, any $p$-pass (randomized) streaming algorithm for the \textsc{$\mathcal{G}$-Editing}$(k)$ problem needs $\Omega(n/p)$ space
\label{thm:fvs:lower}.
\end{theorem}
\begin{proof}

Given an instance of {\sc Disjointness}, we create a graph $G$ which consists of $n$ disjoint copies say $G_1, G_2, \ldots, G_n$ of $H':=H\setminus 2P_2$. By minimality of $H$, it follows that $H'\notin \mathcal{G}$. Let the two edges removed from $H$ to get $H'$ be $e_1$ and $e_2$.  For each $i\in [n]$ we add to $G_i$ the edge $e_1$ iff $x_i=1$ and the edge $e_2$ iff $y_i=1$. Let the resulting graph be $G$.

We now show that $G$ contains a copy of $H$ if and only if {\sc Disjointness} answers YES. Suppose that $G$ contains a copy of $H$. Note that since we add $n$ disjoint copies of $H'$ and add at most two edges ($e_1$ and $e_2$) to each copy, it follows that each connected component of $G$ is in fact a subgraph of $H=H'\cup(e_1 + e_2)$. Since $H$ is connected and $G$ contains a copy of $H$, some connected component of $G$ must exactly be the graph $H$, i.e, to some copy $G_i$ of $H'$ we must have added both the edges $e_1$ and $e_2$. This implies $x_i = 1 = y_i$, and so {\sc Disjointness} answers YES. Now suppose that {\sc Disjointness} answers YES, i.e., there exists $j\in [n]$ such that $x_j = 1= y_j$. Therefore, to the copy $G_j$ of $H'$ we would have added the edges $e_1$ and $e_2$ which would complete it into $H$. So $G$ contains a copy of $H$.

Otherwise due to minimality of $H$, the graph $G$ does not contain any graph from $\mathcal{G}$.
Therefore, any $p$-pass (randomized) streaming algorithm that can determine whether a graph $G$ contains a graph from $\mathcal{G}$ (i.e., answers the question with $k=0$) gives a communication protocol for {\sc Disjointness}, and hence requires $\Omega(n/p)$ space.
\end{proof}

%This implies lower bounds for the following problems:

This implies lower bounds for the following set of problems:
\begin{theorem}
For each of the following problems, any $p$-pass (randomized) algorithm requires $\Omega(n/p)$ space: \textsc{Triangle Packing}$(k)$, \textsc{$s$-Star Packing}$(k)$ and \textsc{Path}$(k)$.
\end{theorem}
\begin{proof}
We first define the problems below:

\begin{center}
\noindent\framebox{\begin{minipage}{0.95\columnwidth}
\textsc{Triangle Packing}$(k)$ \hfill \emph{Parameter}: $k$\\
\emph{Input}: An undirected graph $G=(V,E)$ and an integer $k$\\
\emph{Question}: Do there exist at least $k$ vertex disjoint triangles in $G$?
\end{minipage}}
\end{center}

\begin{center}
\noindent\framebox{\begin{minipage}{0.95\columnwidth}
\textsc{$s$-Star Packing}$(k)$ \hfill \emph{Parameter}: $k$\\
\emph{Input}: An undirected graph $G=(V,E)$ and an integer $k$\\
\emph{Question}: Do there exist at least $k$ vertex disjoint instances of $K_{1,s}$ in $G$ (where $s\geq 3$)?
\end{minipage}}
\end{center}

\begin{center}
\noindent\framebox{\begin{minipage}{0.95\columnwidth}
\textsc{Path}$(k)$ \hfill \emph{Parameter}: $k$\\
\emph{Input}: An undirected graph $G=(V,E)$ and an integer $k$\\
\emph{Question}: Does there exist a path in $G$ of length $\geq k$?
\end{minipage}}
\end{center}
Now we show how each of these problems can be viewed as a \textsc{$\mathcal{G}$-Editing} problem for an appropriate choice of \emph{good} $\mathcal{G}$.

\begin{itemize}
\item \textsc{Triangle Packing}$(k)$ with $k=1$: Take $\mathcal{G}=\{C_3\}$ and $H=C_3$
\item \textsc{$s$-Star Packing}$(k)$ with $k=1$: Take $\mathcal{G}=\{K_{1,s}\}$ and $H=K_{1,s}$
\item \textsc{Path}$(k)$ with $k=3$: Take $\mathcal{G}=\{P_3, P_4, P_5, \ldots\}$ and $H=P_3$
\end{itemize}
We verify the conditions for \textsc{Triangle Packing}$(k)$ with $k=1$; the proofs for other problems are similar. Note that the choice of $\mathcal{G}=\{C_3\}$ and $H=C_3$ implies that $\mathcal{G}$ is \emph{good} since $\mathcal{G}$ only contains one graph. Finally, finding a set of edges $X$ such that the graph $(V, E\cup X)$ contains a graph from $\mathcal{G}$ implies that it has at least one $C_3$, i.e., $X$ is a solution for \textsc{Triangle Packing}$(k)$ with $k=1$.
\end{proof}

\subsubsection{Lower Bound for \textsc{Cluster Vertex Deletion}}
We now show a lower bound for the \textsc{Cluster Vertex Deletion}$(k)$ problem.
\begin{definition}
We say that $G$ is a \emph{cluster} graph if each connected component of $G$ is a clique.
\end{definition}

\begin{center}
\noindent\framebox{\begin{minipage}{0.95\columnwidth}
\textsc{Cluster Vertex Deletion}$(k)$ \hfill \emph{Parameter}: $k$\\
\emph{Input}: An undirected graph $G=(V,E)$ and an integer $k$\\
\emph{Question}: Does there exist a set $X\subseteq V$ of size at most $k$ such that $G\setminus X$ is a cluster graph?
\end{minipage}}
\end{center}

\begin{figure}[t]
\centering
\includegraphics[width=\textwidth]{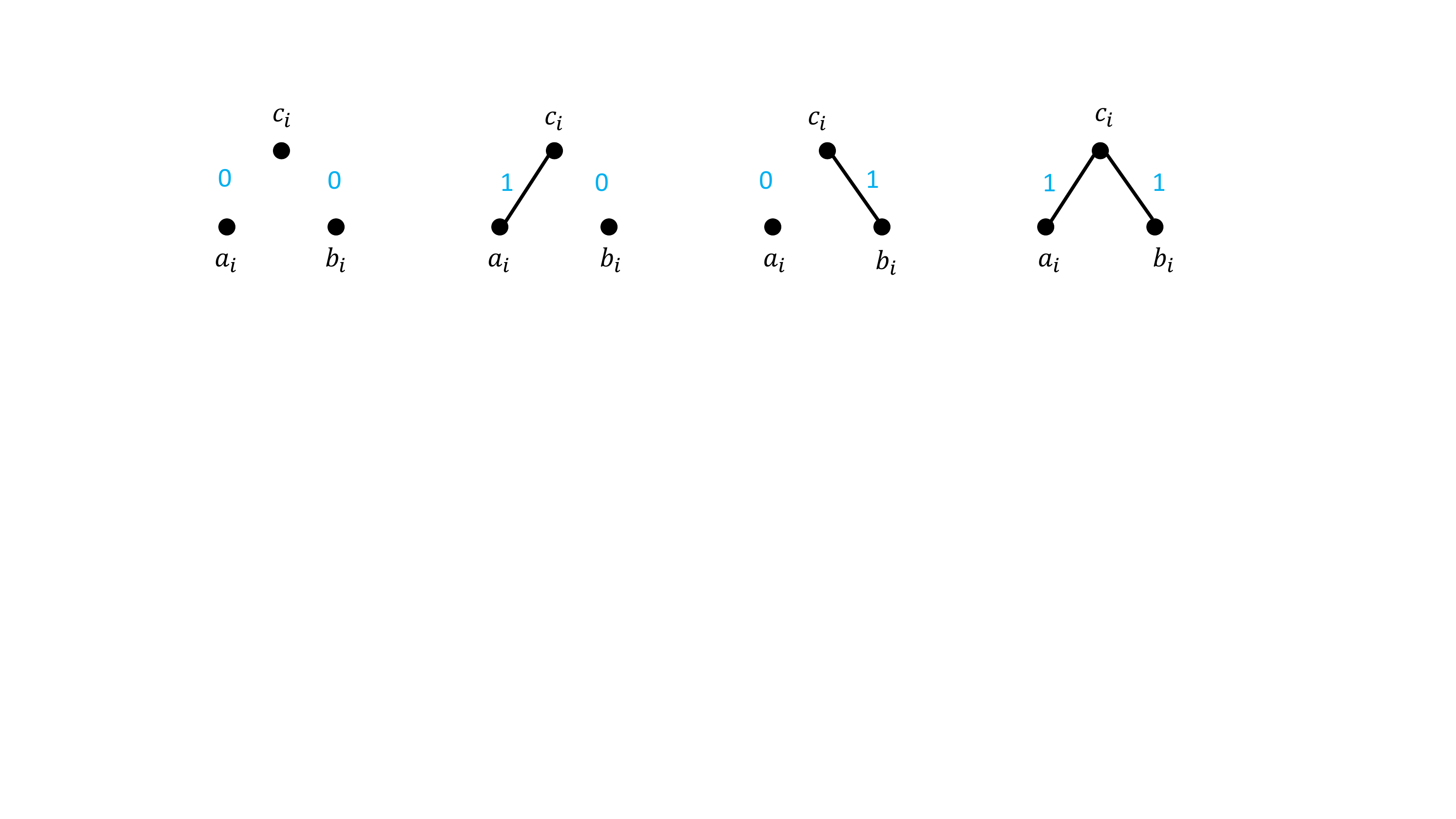}
\vspace{-58mm}
\caption{Gadget for reduction from {\sc Disjointness} to \textsc{Cluster Vertex Deletion}}
\label{fig:fvs}
\vspace{5mm}
\end{figure}

\begin{theorem}
For the \textsc{Cluster Vertex Deletion}$(k)$ problem, any $p$-pass (randomized) streaming algorithm needs $\Omega(n/p)$ space
\label{thm:cvd:lower}.
\end{theorem}
\begin{proof}
Given an instance of {\sc Disjointness}, we create a graph $G$ on $3n$
vertices as follows. For each $i\in [n]$ we create three vertices $a_i, b_i, c_i$. Insert the edge
$(a_i, c_i)$ iff $x_i=1$ and the edge $(b_i, c_i)$ iff $y_i=1$
This is illustrated in Figure~\ref{fig:fvs}.

Now we will show that each connected component of $G$ is a clique iff
{\sc Disjointness} answers NO. In the first direction suppose that
each connected component of $G$ is a clique. Then there cannot exist
$i\in [n]$ such that $x_i = 1 = y_i$ because then the vertices $a_i,
b_i, c_i$ will form a connected component which is a $P_3$; this contradicts  the assumption that each connected component of $G$ is a clique. In reverse direction, suppose that {\sc Disjointness} answers NO. Then it is easy to see that each connected component of $G$ is either $P_1$ or $P_2$, both of which are cliques.

Therefore, any $p$-pass (randomized) streaming algorithm that can determine whether a graph is a cluster graph (i.e., answers the question with $k=0$) gives a communication protocol for {\sc Disjointness}, and hence requires $\Omega(n/p)$ space. \footnote{It is easy to see that the same proof also works for the problems of \textsc{Cluster Edge Deletion}$(k)$ where we can delete at most $k$ edges and \textsc{Cluster Editing}$(k)$ where we can delete/add at most $k$ edges}
\end{proof}

\subsubsection{Lower Bound for \textsc{Minimum Fill-In}}
\label{sec:fill-in}
We now show a lower bound for the \textsc{Minimum Fill-In}$(k)$ problem.
\begin{definition}
We say that $G$ is a \emph{chordal} graph if it does not contain an induced cycle of length $\geq 4$.
\end{definition}

\begin{center}
\noindent\framebox{\begin{minipage}{0.95\columnwidth}
\textsc{Minimum Fill-In}$(k)$ \hfill \emph{Parameter}: $k$ \\
\emph{Input}: An undirected graph $G=(V,E)$ and an integer $k$\\
%\emph{Parameter}: $k$\\
\emph{Question}: Does there exist a set $X$ of at most $k$ edges such that $(V, E\cup X)$ is a chordal graph?
\end{minipage}}
\end{center}

\begin{figure}[t]
\centering
\includegraphics[width=\textwidth]{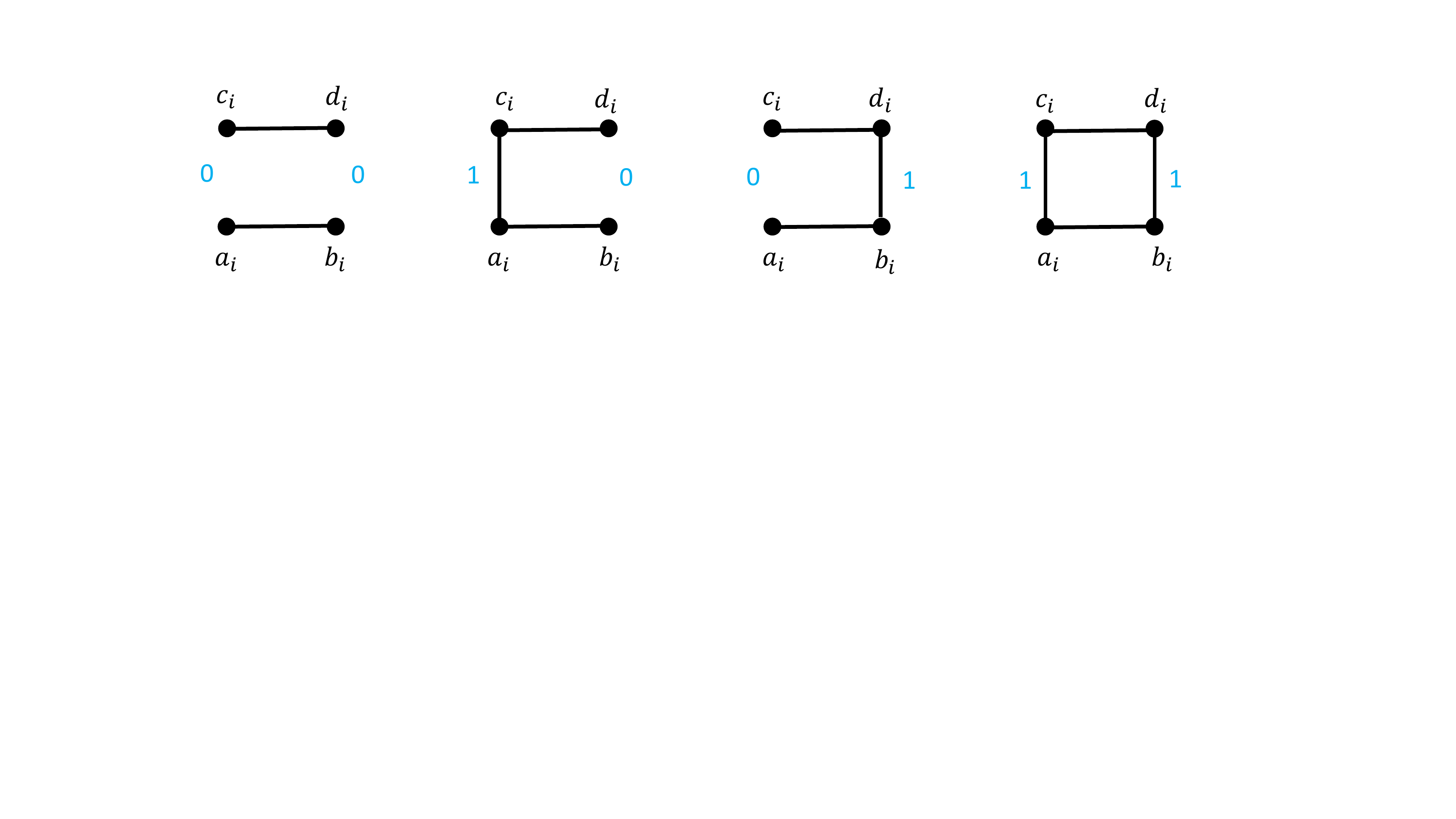}
\vspace{-58mm}
\caption{Gadget for reduction from {\sc Disjointness} to \textsc{Minimum Fill-In}}
\label{fig:cluster}
\vspace{5mm}
\end{figure}

\begin{theorem}
For the \textsc{Minimum Fill-In}$(k)$ problem, any $p$-pass (randomized) streaming algorithm needs $\Omega(n/p)$ space
\label{thm:cvd:lower}.
\end{theorem}
\begin{proof}
We reduce from the {\sc Disjointness} problem in communication complexity. Given an instance of {\sc Disjointness}, we create a graph $G$ on $4n$
vertices as follows. For each $i\in [n]$ we create vertices $a_i, b_i, c_i, d_i$ and insert edges $(a_i, b_i)$ and $(c_i, d_i)$. Insert the edge
$(a_i, c_i)$ iff $x_i=1$ and the edge $(b_i, c_i)$ iff $y_i=1$. This is illustrated in Figure~\ref{fig:cluster}.

Now we will show that $G$ is chordal iff {\sc Disjointness} answers NO. In the first direction suppose that $G$ is chordal. Then there cannot exist $i\in [n]$ such that $x_i = 1 = y_i$ because then the vertices $a_i, b_i, c_i, d_i$ will form an induced $C_4$; contradiction to the fact that $G$ is chordal. In reverse direction, suppose that {\sc Disjointness} answers NO. Then it is easy to see that each connected component of $G$ is either $P_2$ or $P_3$. Hence, $G$ cannot have an induced cycle of length $\geq 4$, i.e., $G$ is chordal.

Therefore, any $p$-pass (randomized) streaming algorithm that can determine whether a graph is a chordal graph (i.e., answers the question with $k=0$) gives a communication protocol for {\sc Disjointness}, and hence requires $\Omega(n/p)$ space.
\end{proof}

\subsubsection{Lower Bound for \textsc{Cograph Vertex Deletion}}
We now show a lower bound for the \textsc{Cograph Vertex Deletion}$(k)$ problem.
\begin{definition}
We say that $G$ is a \emph{cograph} if it does not contain an induced $P_4$.
\end{definition}

\begin{center}
\noindent\framebox{\begin{minipage}{0.95\columnwidth}
\textsc{Cograph Vertex Deletion}$(k)$ \hfill \emph{Parameter}: $k$\\
\emph{Input}: An undirected graph $G=(V,E)$ and an integer $k$\\
\emph{Question}: Does there exist a set $X\subseteq V$ of size at most $k$ such that $G\setminus X$ is a cograph?
\end{minipage}}
\end{center}

\begin{figure}[t]
\centering
\includegraphics[width=\textwidth]{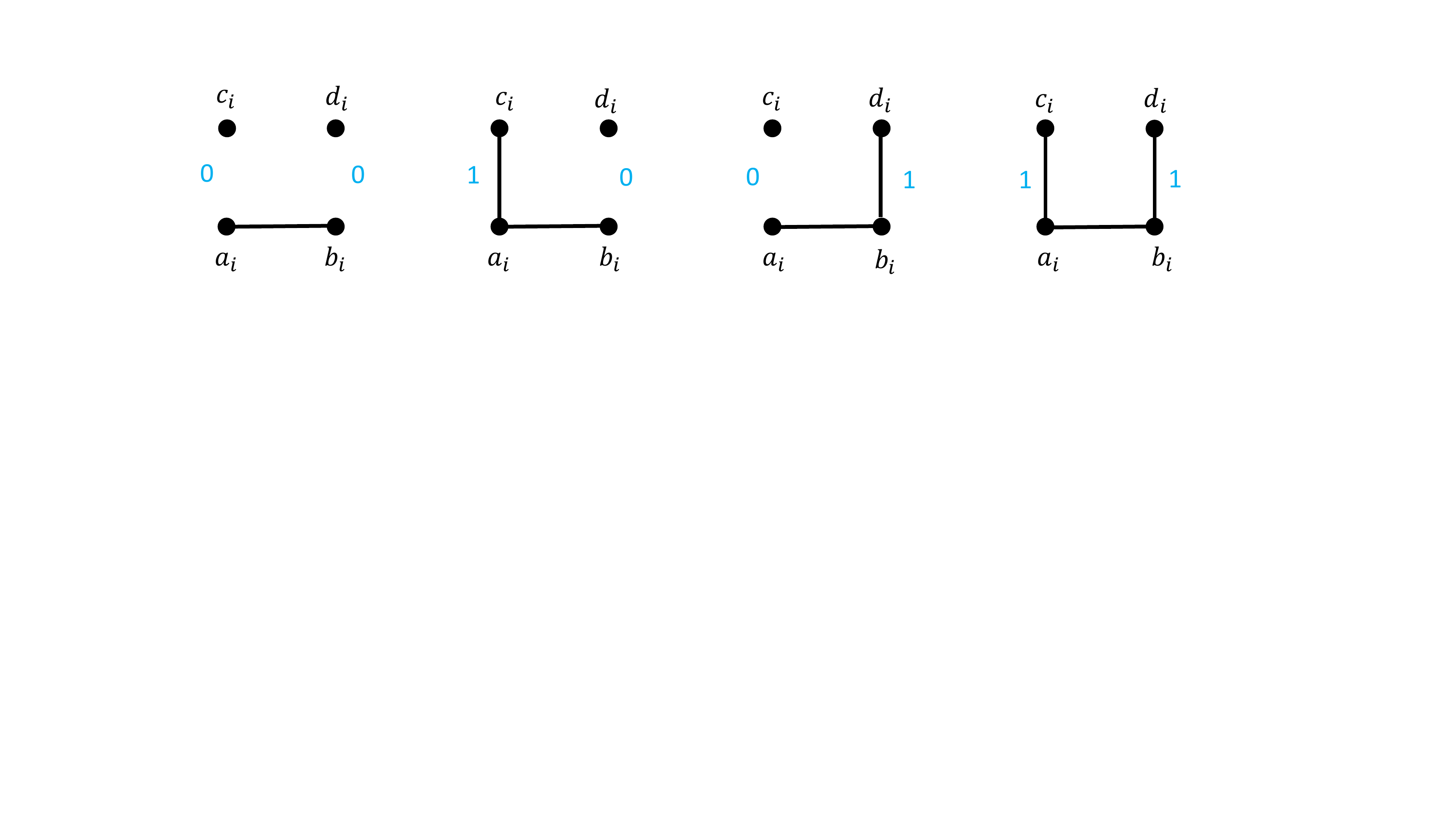}
\vspace{-58mm}
\caption{Gadget for reduction from {\sc Disjointness} to \textsc{Cograph Vertex Deletion}}
\label{fig:covd}
\vspace{5mm}
\end{figure}

\begin{theorem}
For the \textsc{Cograph Vertex Deletion}$(k)$ problem, any $p$-pass (randomized) streaming algorithm needs $\Omega(n/p)$ space
\label{thm:covd:lower}.
\end{theorem}
\begin{proof}
We reduce from the {\sc Disjointness} problem in communication complexity. Given an instance of {\sc Disjointness}, we create a graph $G$ on $4n$
vertices as follows. For each $i\in [n]$ we create vertices $a_i, b_i, c_i, d_i$ and insert edges $(a_i, b_i)$. Insert the edge
$(a_i, c_i)$ iff $x_i=1$ and the edge $(b_i, c_i)$ iff $y_i=1$. This is illustrated in Figure~\ref{fig:covd}.

Now we will show that $G$ has an induced $P_4$ if and only if {\sc Disjointness} answers YES. In the first direction suppose that $G$ has an induced $P_4$. Since each connected component of $G$ can have at most 4 vertices, it follows that the $P_4$ is indeed given by the path $c_{i}-a_{i}-b_{i}-d_{i}$ for some $i\in [n]$. By construction of $G$, this implies that $x_i = 1 = y_i$, i.e., {\sc Disjointness} answers YES. In reverse direction, suppose that {\sc Disjointness} answers YES. Then there exists $j\in [n]$ such that the edges $(a_{i}, c_{i})$ and $(b_i, d_i)$ belong to $G$. Then $G$ has the following induced $P_4$ given by $c_{j}-a_{j}-b_{j}-d_{j}$.

Therefore, any $p$-pass (randomized) streaming algorithm that can determine whether a graph is a cograph (i.e., answers the question with $k=0$) gives a communication protocol for {\sc Disjointness}, and hence requires $\Omega(n/p)$ space.
\end{proof}

\subsubsection{\textsc{Bipartite Colorful Neighborhood}}

%We show a lower bound for the problem of finding whether there exists an odd cycle transversal of size at most $k$:
We now show a lower bound for the \textsc{Bipartite Colorful Neighborhood}$(k)$ problem.

\begin{center}
\noindent\framebox{\begin{minipage}{0.95\columnwidth}
\textsc{Bipartite Colorful Neighborhood}$(k)$ \hfill \emph{Parameter}: $k$\\
\emph{Input}: A bipartite graph $G=(A, B,E)$ and an integer $k$\\
\emph{Question}: Is there a 2-coloring of $B$ such that there exists a set $S\subseteq A$ of size at least $k$ such that each element of $S$ has at least one neighbor in $B$ of either color?
\end{minipage}}
\end{center}

\begin{figure}%[t]
\centering
\includegraphics[width=\textwidth]{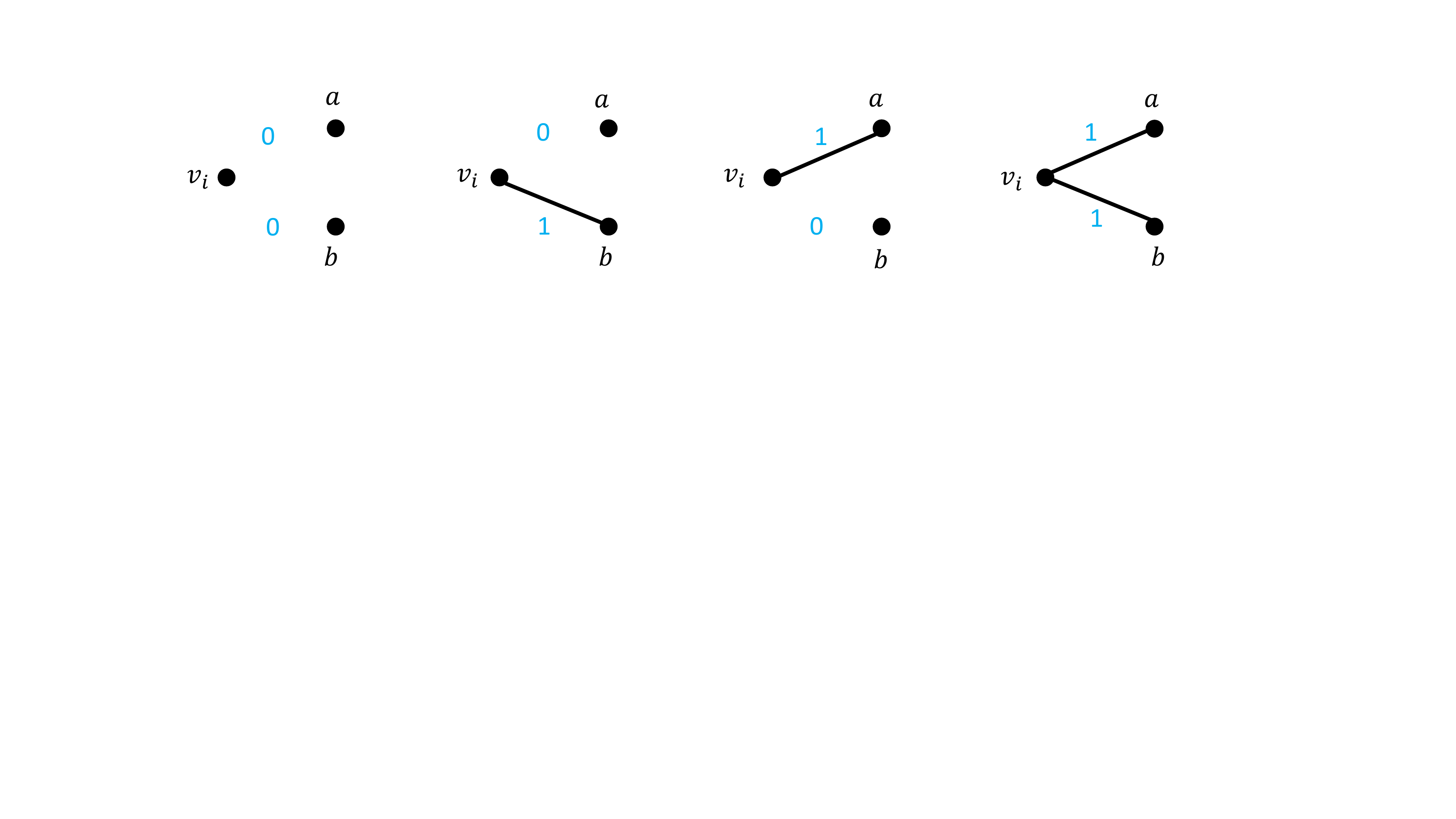}
\vspace{-58mm}
\caption{Gadget for reduction from {\textsc Disjointness} to \textsc{Bipartite Colorful Neighborhood}}
\label{fig:bip-color}
\vspace{5mm}
\end{figure}

\begin{theorem}
For the \textsc{Bipartite Colorful Neighborhood}$(k)$ problem, any $p$-pass (randomized) streaming algorithm needs $\Omega(n/p)$ space
\label{thm:fvs:lower}.
\end{theorem}
\begin{proof}
We reduce from the {\sc Disjointness} problem in communication complexity.
Given an instance of {\sc Disjointness}, we create a graph $G$ on $n+2$
vertices as follows. For each $i\in [n]$ we create a vertex $v_i$. In addition, we have two special vertices $a$ and $b$. For each $i\in [n]$, insert the edge
$(a, v_i)$ iff $x_i=1$ and the edge $(b, v_i)$ iff $y_i=1$. Let $A=\{v_1, v_2, \ldots, v_n\}$ and $B=\{a,b\}$. This is illustrated in Figure~\ref{fig:bip-color}.

Now we will show that $G$ answers YES for \textsc{Bipartite Colorful Neighborhood}$(k)$ with $k=1$ iff {\sc Disjointness} answers YES. In the first direction suppose that $G$ answers YES for \textsc{Bipartite Colorful Neighborhood}$(k)$ with $k=1$. Let $v_i$ be the element in $A$ which has at least one neighbor in $B$ of either color. Since $|B|=2$, this means that $v_i$ is adjacent to both $a$ and $b$, i.e., $x_i = 1 = y_i$ and hence {\sc Disjointness} answers YES. In reverse direction, suppose that {\sc Disjointness} answers YES. Hence, there exists $j\in [n]$ such that $x_j = 1= y_j$. This implies that $v_j$ is adjacent to both $a$ and $b$. Consider the 2-coloring of $B$ by giving different colors to $a$ and $b$. Then $S=\{v_j\}$ satisfies the condition of having a neighbor of each color in $B$, and hence $G$ answers YES for \textsc{Bipartite Colorful Neighborhood}$(k)$ with $k=1$.

Therefore, any $p$-pass (randomized) streaming algorithm that can solve \textsc{Bipartite Colorful Neighborhood}$(k)$ with $k=1$ gives a communication protocol for {\sc Disjointness}, and hence requires $\Omega(n/p)$ space.

\end{proof}

%\section{Pseudocode of Algorithm {\sc RandomPartition}}
%\label{sec:pseudo:part}
%\input{8-code}

\end{document}